\documentclass[11pt]{article}

\usepackage[utf8]{inputenc}
\usepackage{fullpage}
\usepackage{titling}
\usepackage{graphicx}

\usepackage{natbib}
\bibpunct{(}{)}{,}{a}{}{;}

\usepackage[page]{appendix}
\usepackage{amsmath}
\usepackage{amssymb}
\usepackage{amsthm}

\usepackage[hidelinks]{hyperref}

\title{Test \& Roll: Profit-Maximizing A/B Tests}

\author{Elea McDonnell Feit \\
LeBow College of Business \\ 
Drexel University \\ 
eleafeit@gmail.com 
\and 
Ron Berman \\ 
The Wharton School \\ 
University of Pennsylvania \\ 
ronber@wharton.upenn.edu }

\date{\today}

\begin{document}

\maketitle
\thispagestyle{empty} 
\vspace{3in}
\noindent We thank Julie Albers, Eduardo Azevedo, Eric Bradlow, Raghu Iyengar, Pete Fader, Bruce McCullough, and Christophe Van den Bulte for helpful discussions and suggestions and we gratefully acknowledge support from a 2017 Adobe Digital Experience Research Award.

\newpage
\pagenumbering{arabic} 

\begin{center}
\Large
\thetitle
\end{center}
\vspace{0.5in}

\begin{abstract}
Marketers often use A/B testing as a tool to compare marketing treatments in a test stage and then deploy the better-performing treatment to the remainder of the consumer population. While these tests have traditionally been analyzed using hypothesis testing, we re-frame them as an explicit trade-off between the opportunity cost of the test (where some customers receive a sub-optimal treatment) and the potential losses associated with deploying a sub-optimal treatment to the remainder of the population. 

We derive a closed-form expression for the profit-maximizing test size and show that it is substantially smaller than typically recommended for a hypothesis test, particularly when the response is noisy or when the total population is small. The common practice of using small holdout groups can be rationalized by asymmetric priors. The proposed test design achieves nearly the same expected regret as the flexible, yet harder-to-implement multi-armed bandit under a wide range of conditions. 

We demonstrate the benefits of the method in three different marketing contexts -- website design, display advertising and catalog tests -- in which we estimate priors from past data. In all three cases, the optimal sample sizes are substantially smaller than for a traditional hypothesis test, resulting in higher profit.

\end{abstract}

\newpage

\section{Introduction}

Experimentation is an important tool for marketers in a wide range of settings including direct mail, email, display advertising, social media marketing, website optimization, and app design. In tactical marketing settings, which we call ``test \& roll'' experiments, data on customer response is first collected in a test stage where a subset of customers are randomly assigned to a treatment. In the roll stage that follows, marketers deploy one treatment to all remaining customers based on the test results.

Figure \ref{fig:test_setup} shows an example test \& roll setup screen. Emails with two different subject lines will each be sent to 8,910 customers at random from a total list of 59,404 email addresses. Once the test outcomes are measured, the platform sends the better-performing email to the remainder of the list. 
\begin{figure}[ht]
\centering
\includegraphics[width=0.7\textwidth]{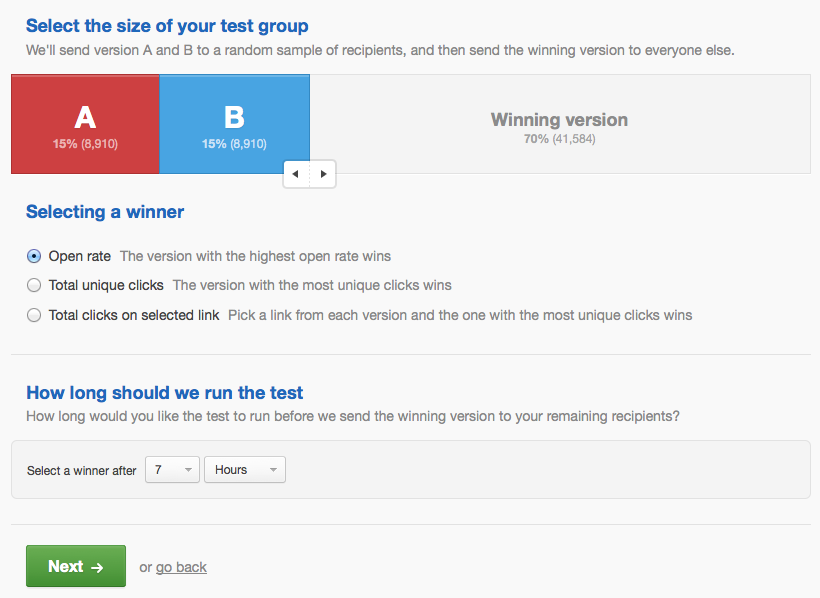}
\caption{Typical test \& roll setup. (Screenshots from the email marketing tool Campaign Monitor, as described on the \href{https://zapier.com/learn/email-marketing/ab-testing-email-marketing/}{Zapier.com blog}.)}
\label{fig:test_setup}
\end{figure}

Traditionally, such randomized controlled trials are analyzed with a significance test, where the null hypothesis of equal mean response of two treatments is rejected if
\begin{equation} 
\overline{y}_1 - \overline{y}_2 \geq z_{1-\alpha/2} \sqrt{\frac{s_1^2}{n_1} + \frac{s_2^2}{n_2}}
\label{eq:null_hypothesis_test}
\end{equation}
\noindent where $\overline{y}_1$ and $\overline{y}_2$ are the mean response for each test group, $s_1$ and $s_2$ are the standard deviation of the response, $n_1$ and $n_2$ are the sample sizes, and the significance level $\alpha$ is the desired type I error rate that determines the critical value $z$.\footnote{We focus on a $z$-test for simplicity. The test of proportions is similar.}

When using hypothesis testing, the sample size is fixed prior to data collection and $n_1$ and $n_2$ are set to detect an effect of at least $d$ with probability $1 - \beta$. When $s_1 = s_2 = s$, the recommended sample size is: 
\begin{equation}
n_{HT}=n_1 = n_2 \approx (z_{1-\alpha/2} + z_\beta)^2 \left( \frac{2 s^2}{d^2} \right)\\
\label{eq:nht_sample_size}
\end{equation}
The recommendation is to set $n_1 = n_2$ because this maximizes the statistical power of the experiment when $s_1 = s_2$. 
 
We develop an alternative approach to planning and analyzing A/B tests with finite populations. While null hypothesis testing is the ``gold standard'' in scientific and medical research and is often recommended for marketing tests \citep[e.g.,][]{PekelisWalshJohari2015WP}, the statistical significance threshold in \eqref{eq:null_hypothesis_test} is a poor decision rule for test \& roll experiments aimed at maximizing profits, for four reasons.

First, hypothesis tests at typical significance levels (e.g., $\alpha = 0.05$) are designed to avoid concluding that two treatments perform differently when they do not. Yet these Type I errors have little consequence for profit, assuming no deployment costs. If the null can not be rejected and both treatments yield identical effects, the same profit will be earned regardless of which treatment is deployed. Because of the profit trade-off between test-stage learning and roll-stage earning, conservative sample sizes based on null hypothesis testing lower overall expected profit, by exposing too many people to the less effective treatment in the test. 

Second, the population available for testing and deploying is often limited, but the recommended sample size in \eqref{eq:nht_sample_size} does not take this constraint into account. In online advertising experiments where effects are often small (but profitable), the recommended sample size may be larger than the size of the population itself \citep{LewisRao2015}.\footnote{\label{eq:nht_fpc_sample_size} The seldom-used finite population correction (FPC) will recommend sample sizes smaller than the population, however this correction does not account for the opportunity cost of the test and does not maximize profit. The FPC adjusts the standard error to correct for the inaccuracy of the central limit theorem when sampling from a finite population of size $N$ without replacement resulting in a recommended sample size of 
$n_{FPC} = \frac{2 N (z_{1-\alpha/2}+z_{\beta})^2 s^2}{ (N-1) d^2 + 4 s^2 (z_{1-\alpha/2}+z_{\beta})^2}
$. We thank an anonymous reviewer for suggesting this as a benchmark.} Yet, as we show, when the population is limited, smaller tests that will never reach statistical significance can still have substantial benefit in improving expected profit. 

Third, the typical null hypothesis test in \eqref{eq:null_hypothesis_test} provides no guidance on which treatment to deploy when the results are not significant. Many A/B testers advocate deploying the incumbent treatment (if there is one) in the interest of being ``conservative",  choosing randomly \citep{Gershoff2017}, or continuing the test until it reaches significance \citep[e.g.,][]{berman2018}. 

Fourth, practitioners often design tests with unequal sample sizes for each treatment \citep[e.g.,][]{LewisRao2015, ZantedeschiFeitBradlow2017}. Our framework allows unequal sample sizes to arise naturally from prior beliefs, whereas this can not be rationalized under classical hypothesis testing  when response variance is equal ($s_1^2 = s_2^2$).

We re-frame the test \& roll decision problem in Section \ref{sec:model}, focusing on profit and making an explicit trade-off between the opportunity cost of the test (where some customers receive the sub-optimal treatment) and the losses associated with deploying the sub-optimal treatment to the remainder of the finite population. Effectively, the problem we define can be seen as a constrained version of a multi-armed bandit, where there are only two allocation decisions instead of many.

We derive a new closed-form expression for the profit-maximizing sample size in Section \ref{sec:normal}, assuming that the average revenue per customer is normally-distributed with normal priors. Test sample sizes under this framework are often substantially smaller than those recommended by (\ref{eq:nht_sample_size}). Unlike sample sizes for a hypothesis test that increase linearly with the \emph{variance} of the response in \eqref{eq:nht_sample_size}, profit-maximizing sample sizes increase sub-linearly with the \emph{standard deviation} of the response, leading to substantially smaller test sizes when the response is noisy. Profit-maximizing samples are also proportional to the square root of the total size of the population available, and so they naturally scale to both large and small settings.

Improved performance is achieved because profit-maximizing tests identify the best performing treatment with high probability when treatment effects are large; the lost profit (regret) from errors in treatment selection is small when treatment effects are small. We also show that a test \& roll with the profit-maximizing sample size achieves nearly the same level of regret as the popular Thompson sampling solution to the multi-armed bandit problem \citep{scott2010,schwartzBradlowFader2016WP}; both have regret of $O(\sqrt{N})$. Although sub-optimal relative to a multi-armed bandit, the profit-maximizing test \& roll provides a transparent decision point and reduced operational complexity without significant loss of profit. 

Section \ref{sec:asymmetric} extends the analysis to situations with different priors on treatments, and provides an efficient numeric approach to computing optimal sample sizes. This allows us to rationalize the common practice of using unequally-sized treatment groups when the two treatments are believed \emph{a priori} to produce different responses, e.g., a test comparing media exposure to no exposure or a test comparing two different prices. 

To illustrate how test \& roll experiments should be designed in practice, Section \ref{sec:applications} provides three empirical applications: website design, online display advertising, and catalog marketing. For each application, we estimate priors based on previous similar experiments. These applications show the wide range of test designs that result from different priors and show that the ``one-size-fits-all" approach favored by null hypothesis testing does not maximize profit. We conclude in Section \ref{sec:discussion} with a discussion of potential extensions of the test \& roll framework and implications for A/B testers. Full statements of propositions and proofs appear in the Appendix.

\section{The Test \& Roll Decision Problem}
\label{sec:model}

% Begin with to 1/2 notation here 
A test \& roll with a population of $N$ customers has two stages: a test stage and a roll stage. In the test stage, a random sample of $n_1$ customers are exposed to treatment 1 and a random non-overlapping sample of $n_2$ customers are exposed to treatment 2, with $n_1+n_2 < N$. In the roll stage, all remaining $N-n_1-n_2$ customers receive either treatment 1 or treatment 2 based on a decision rule that incorporates the data observed in the test stage. The marketer's goal is to maximize the cumulative profit earned in both stages. 

Assuming the profit for each customer receiving treatment $j$ is an independent random variable $Y_j$ that follows a distribution with parameters $\theta_j$, the expected profit earned during the test phase is:  
\begin{equation}
E_{Y_1,Y_2}[\Pi_{\textnormal{T}}|\theta_1, \theta_2] = n_1 E[Y_1|\theta_1] + n_2 E[Y_2|\theta_2]
\label{eq:exp_profit_test_conditional}
\end{equation}
$Y_j$ is the profit net of any costs related to the treatments, e.g., media costs or discounts. In website and email tests, for example, the cost of both treatments is the same and can be ignored. 

Denote the vector of observed profit from customers exposed to treatment $j$ in the test as $y_j = y_{j,1}, \ldots, y_{j,n_j}$. Once $y_1$ and $y_2$ are observed, the analyst chooses a treatment to deploy with the remaining $N - n_1 - n_2$ customers. Let $\delta(y_1, y_2)$ be the decision rule which takes the value 1 for the decision to deploy treatment 1 and 0 for treatment 2. The optimal decision rule is to select the treatment with the highest posterior predictive mean $E[Y_j|y_j]$ \citep{DeGroot1970}.

Depending on the decision rule, the expected profit in the roll stage is:
\begin{equation}
E_{Y_1,Y_2}[\Pi_{\textnormal{D}}|\theta_1, \theta_2] = (N-n_1-n_2) E_{Y_1,Y_2}[\delta(y_1,y_2)Y_1+(1-\delta(y_1,y_2))Y_2|\theta_1,\theta_2]
\label{eq:exp_profit_deploy_conditional}
\end{equation}

Increasing $n_1$ and $n_2$ provides more observations about the profitability of each treatment, and thus has the potential to yield more correct decisions in the roll stage. Simultaneously, increasing $n_1$ and $n_2$ decreases the population remaining in the roll stage and increases the test population, some of which is exposed to the lesser-performing treatment. Thus, the test \& roll framework sets up an explicit trade-off between learning during the test phase and earning during the roll phase. This trade-off is important when the total population size $N$ is limited.\footnote{Treatments may be defined as a single exposure (e.g., an email) or a series of exposures (e.g., a digital media campaign).} Well-defined, limited populations are common in marketing: in direct marketing $N$ is the size of the customer list \citep{bitran1996,BonfrerDreze2009}; in paid media, $N$ is often determined by a finite budget; in website or app tests, $N$ reflects the expected traffic for some period after the test.

The parameters $\theta_1$ and $\theta_2$ are unknown prior to the test (hence the need for the test). By assuming a prior distribution over these parameters, we obtain the \emph{a priori} expected profit of the A/B test:
\begin{equation}
    E_{\theta_1,\theta_2}[E_{Y_1,Y_2}[\Pi_{\textnormal{T}}|\theta_1, \theta_2]+E_{Y_1,Y_2}[\Pi_{\textnormal{D}}|\theta_1, \theta_2]]
\end{equation}

Designing the test entails selecting the sample sizes $n_1$ and $n_2$ that maximize the total expected profit:
\begin{equation}
(n_1^*, n_2^*) = \underset{n_1, n_2}{\mathrm{argmax}}~ E_{\theta_1,\theta_2}[E_{Y_1,Y_2}[\Pi_{\textnormal{T}}|\theta_1, \theta_2]+E_{Y_1,Y_2}[\Pi_{\textnormal{D}}|\theta_1, \theta_2]]
\label{eq:sample_size_prob}
\end{equation}

Thus a profit-maximizing test \& roll runs a test with the sample size in \eqref{eq:sample_size_prob} and deploys one treatment based on the decision rule $\delta$.

Both our approach and the hypothesis testing approach described in equations \eqref{eq:null_hypothesis_test} and \eqref{eq:nht_sample_size} are decision-theoretic but differ in three aspects: (1) We define the decision as whether to deploy treatment 1 or treatment 2, instead of deciding whether to reject the null hypothesis; (2) The objective in hypothesis testing is to maximize statistical power while controlling type-I error, while we focus on maximizing profits; (3) Hypothesis testing uses a 0/1 loss function, and so every incorrect decision has the same cost, while our approach uses the actual opportunity cost as the loss, including the cost of the test.

Similar two-stage decision problems have appeared in the literature. \citet{chick2001} analyze a two-stage decision problem where the cost of the test is a fixed multiple of the sample sizes, rather than actual opportunity cost as we have here. In studying multi-armed bandits \citet{schwartzBradlowFader2016WP} and \citet{misra2019} use a test \& roll as a benchmark, but they do not optimize the sample size. The closest work comes from the clinical trials literature, where \citet{cheng2003} defines the same test \& roll problem with a finite ``patient horizon'' and approximates the optimal sample size for Bernoulli responses with beta priors. \citet{stallard2017} extends \citet{cheng2003} to exponential family responses with conjugate exponential family priors. As a result, they also need to use approximations to compute the optimal sample size. In this paper, we focus on Normal response distributions with Normal priors, which allows us to provide an exact closed-form for the optimal sample size as well as exact expected profit and regret, which we show next.  

\section{Test \& Roll with Symmetric Normal Priors}
\label{sec:normal}

To derive a profit-maximizing sample size formula, we assume $Y_1 \sim \mathcal{N}(m_1, s^2)$ and $Y_2 \sim \mathcal{N}(m_2, s^2)$ with identical priors $m_1, m_2 \sim \mathcal{N}(\mu, \sigma^2)$. The variance of the response, $s^2$ is known; in practice it can be estimated from previously observed responses.\footnote{The assumption that $s_1$ and $s_2$ are known could easily be relaxed by putting priors on them, but this is not necessary for deriving key insights.} The hyper-parameters $\mu$ and $\sigma$ represent expectations for how the two treatments may perform, which can be be informed by previous similar marketing campaigns (as illustrated in Section \ref{sec:applications}). 

The symmetric priors imply that neither treatment is \emph{a priori} likely to perform better, but they do not imply that $m_1 = m_2$. The implied prior on the treatment effect $m_1 - m_2$ is $\mathcal{N}(0, 2\sigma^2)$ and the absolute difference between treatments $|m_1 - m_2|$ is distributed half-normal with mean $\frac{\sqrt{2}}{\sqrt{\pi}} \sigma$. Thus $\sigma$ is related to the \emph{a priori} expectation about the potential difference in treatment effects (as well as the uncertainty).

The expected profit in the test stage for this model is: 
\begin{equation} 
E[\Pi_{\textnormal{T}}] = (n_1 + n_2) \mu
\label{eq:NN_exp_profit_test}
\end{equation}

The expected profit in the roll stage depends on the decision rule $\delta(y_1, y_2$). The profit-maximizing decision rule is to choose the treatment with the greater expected posterior mean response: 
\begin{equation}
\delta(y_1, y_2)  = I\left( \left( \frac{1}{\sigma^2} + \frac{n_1}{s^2} \right)^{-1} \left( \frac{\mu}{\sigma^2} + \frac{\overline{y}_1}{s^2} \right) > \left( \frac{1}{\sigma^2} + \frac{n_2}{s^2} \right)^{-1} \left( \frac{\mu}{\sigma^2} + \frac{\overline{y}_2}{s^2} \right) \right)
\label{eq:decision_rule_post_mean}
\end{equation}
where $\overline{y}_j$ is the average response observed for treatment $j$ and $I(\cdot)$ is the indicator function.
Since the priors are symmetric, this reduces to $\delta(y_1, y_2) = I\left( \overline{y}_1 > \overline{y}_2 \right)$ if $n_1 = n_2$, i.e., the highly-intuitive ``pick the winner" in the test.

Proposition \ref{prop:exp-prof} shows that the decision rule in \eqref{eq:decision_rule_post_mean} yields an expected roll-stage profit of: 
\begin{equation}
E[\Pi_{\textnormal{D}}] = (N - n_1 - n_2) \left[ \mu + \frac{\sqrt{2} \sigma^2}{\sqrt{\pi} \sqrt{2\sigma^2 + \frac{n_1 + n_2}{n_1 n_2}s^2} }\right]
\label{eq:NN_exp_profit_deploy_symmetric}
\end{equation}
The second addend in the square brackets is the expected incremental profit per customer earned by (usually) deploying the better treatment relative to choosing randomly with expected profit of $\mu$. Unsurprisingly, the incremental gain per customer from the test is increasing in the sample sizes $n_1$ and $n_2$. However, as $(n_1 + n_2)$ increases, the number of customers for whom this higher profit is earned is smaller.  The incremental gain decreases with the noise in the data, $s$, as expected. The \emph{a priori} range of effect sizes is defined by $\sigma$. Higher \emph{a priori} uncertainty about the mean response increases the option value from the experiment and so the incremental gain increases with $\sigma$. 

To find the optimal sample size, the sum of the test profit in (\ref{eq:NN_exp_profit_test}) and the deployment profit in (\ref{eq:NN_exp_profit_deploy_symmetric}) can be maximized over $n_1$ and $n_2$ resulting in optimal sample sizes (Proposition \ref{prop:sample-size}):
\begin{equation}
n^* = n_1^* = n_2^* %= \frac{\sqrt{9\sigma^4 + 4 n \sigma^2 \sigma_0^2} - \frac{3}{4} \sigma^2}{4\sigma_0^2} 
= \sqrt{\frac{N}{4}\left( \frac{s}{\sigma} \right)^2 + \left( \frac{3}{4} \left( \frac{s}{\sigma} \right)^2  \right)^2 } -  \frac{3}{4} \left(\frac{s}{\sigma} \right)^2
\label{eq:NN_sample_size}
\end{equation}
Since $\frac{N}{4}\left( \frac{s}{\sigma} \right)^2 + \left( \frac{3}{4} \left( \frac{s}{\sigma} \right)^2  \right)^2 \le \left(\sqrt{N} \frac{s}{2 \sigma}+ \frac{3}{4} \left( \frac{s}{\sigma} \right)^2\right)^2$, \eqref{eq:NN_sample_size} implies that $n_1^* = n_2^* \le \sqrt{N} \frac{s}{2 \sigma}$. The profit-maximizing sample size is always less than the population size $N$ and grows sub-linearly with the standard deviation of the response $s$. By contrast, the recommended sample size for a hypothesis test in \eqref{eq:nht_sample_size} grows linearly with the variance $s^2$ without regard to $N$. This explains why, for noisy responses, hypothesis tests frequently require sample sizes that are larger than the available population \citep{LewisRao2015}. 

Notably, the profit-maximizing sample size decreases with $\sigma$. Large $\sigma$ implies: (1) a larger expected difference between treatments and, (2) a lower error rate for a given sample size (see \eqref{eq:NN_error_rate} below), while (3) the opportunity cost remains the same. 

\subsection{Error rate}

Test \& roll does not require the planner to specify an acceptable level of error; the error rate follows from optimally trading off the opportunity cost of the test against the expected loss in profit due to deployment errors. However, practitioners may want to know the expected error rate. Conditional on $m_1$ and $m_2$, the likelihood of deploying treatment 1 when treatment 2 has a better mean response is: 
\begin{equation}
Pr\left(\delta(y_1, y_2) = 1 | m_1, m_2\right) =  1-\Phi\left(\frac{m_2-m_1}{s \sqrt{\frac{1}{n_1}+\frac{1}{n_2}}}\right)
\label{eq:NN_error_rate_conditional}
\end{equation}
From (\ref{eq:NN_error_rate_conditional}), we see that when the difference in treatments $m_2 - m_1$ is positive and large, the error rate is lower, i.e., the better treatment will be deployed. When $m_2 - m_1$ is smaller, it is more likely that the wrong treatment will be deployed, but this is less consequential for profit. 

Integrating (\ref{eq:NN_error_rate_conditional}) over the priors on $m_1$ and $m_2$, the expected error rate is (Corollary \ref{cor:error-rate}): 
\begin{align}
E[Pr(\delta(y_1, y_2) = 1 & | m_1 < m_2)] = E[Pr(\delta(y_1, y_2) = 0 | m_1 > m_2)] = \nonumber \\ &\frac{1}{4} - \frac{1}{2 \pi}\arctan\left(\frac{\sqrt{2}\sigma}{s} \sqrt{\frac{n_1 n_2}{n_1+n_2}}\right)
\label{eq:NN_error_rate}
\end{align}
As expected, the error rate decreases with the test sizes $n_1$ and $n_2$, increases with $s$, and decreases with $\sigma$.

\subsection{Regret}
\label{sec:regret}

To provide an upper bound on the total expected profit, we compute the expected profit with perfect information (PI). If an omniscient marketer were able to deploy the treatment with higher expected profit to all $N$ customers without testing, the expected profit would be (Proposition \ref{prop:regret}, part 1): 
\begin{equation}
E[\Pi|\textnormal{PI}] = \left(\mu + \frac{\sigma}{\sqrt{\pi}}\right) N
\label{eq:exp_profit_perfect_information}
\end{equation}
The expected profit of any algorithm for choosing which treatment to deploy to each customer will be between the expected value of choosing randomly, which is $\mu N$ and the expected value of perfect information in (\ref{eq:exp_profit_perfect_information}). The expected profit with perfect information scales with the variance of the prior $\sigma$; the more potential difference there is between treatments, the more opportunity there is to improve profits by choosing the better treatment. 

The expected regret of the profit-maximizing test \& roll experiment is (Proposition \ref{prop:regret}, part 2):
\begin{align}\label{eq:exp-regret}
    E[\Pi|PI] - E[\Pi^*_D+\Pi^*_T] = & N \frac{\sigma}{\sqrt{\pi}}\left(1-\frac{\sigma}{\sqrt{\sigma^2+\frac{s^2}{n^*}}}\right)+\frac{2n^*\sigma^2}{\sqrt{\pi}\sqrt{\sigma^2+\frac{s^2}{n^*}}} \nonumber \\
    & \le \frac{3s\sqrt{N}}{\sqrt{\pi}}= O(\sqrt{N})
\end{align}
When populations are larger, the regret per customer decreases, hence marketers with larger populations have a greater opportunity to improve profits on a per-customer basis with a profit-maximizing test. We also see that the regret has an upper bound that does not depend on $\sigma$, implying that the potential regret is limited. Marketers can use the new closed-form formulas in \eqref{eq:NN_error_rate} and  \eqref{eq:exp-regret} to easily assess the potential value of running a test \& roll.

To gain further insight into these results, we look at the expected relative regret of the profit-maximizing test \& roll with respect to the expected profit from perfect information:
\begin{equation} 
\frac{E[\Pi|PI] - E[\Pi^*_D+\Pi^*_T]}{E[\Pi|PI]}
\end{equation}
 As corollary \ref{cor:max-regret} proves, the relative regret reaches a maximum for an intermediate finite value of $\sigma$. When $\sigma$ is very small, there is not much to gain from having perfect information, and hence the relative regret will be small, while when $\sigma$ is large, the test stage will pick the best performing treatment with a very high probability, also yielding low-regret. Only when $\sigma$ is intermediate is there some chance of substantial loss from using a simple method such as test \& roll, but even in this case the potential loss is limited.

In contrast, using the sub-optimal sample size recommended for a hypothesis test produces substantially greater regret. Assuming that the better performing treatment will be deployed after the test regardless of significance,\footnote{As noted in the introduction, it is not clear what should be done if the null hypothesis cannot be rejected.} we can substitute the value of $n_{HT}$ from \eqref{eq:nht_sample_size} for $n^*$ in \eqref{eq:exp-regret}. The regret from using the larger sample size is (Proposition \ref{prop:regret}, part 3):
\begin{equation}
E[\Pi|PI] - E[\Pi_D+\Pi_T|HT] \ge N \frac{\sigma}{\sqrt{\pi}}\frac{d^2}{4(z_{(1-\alpha)/2}+z_{\beta})^2\sigma^2+2d^2)} = \Omega(N)
\end{equation} 
implying that hypothesis testing has a lower bound expected regret of $\Omega(N)$, substantially larger than the profit-maximizing sample size with regret $O(\sqrt{N})$ as $N$ becomes large.\footnote{$\Omega$ denotes a lower asymptotic bound while $O$ denotes an upper asymptotic bound.} Proposition \ref{prop:regret} also shows that this bound holds when a finite-population correction is included in the sample size formula. 

We can also compare a test \& roll with profit-maximizing sample size to a multi-armed bandit where allocation to treatments is determined probabilistically for each customer based on previous responses. \cite{agrawal2013} show that the expected regret of a multi-armed bandit with Thompson sampling \citep{Thompson1933}\footnote{Thompson sampling uses a decision rule based on the posterior similar to \eqref{eq:decision_rule_post_mean}, but continuously updates the posterior and makes a probabilistic decision for each customer proportional to the probability that each treatment is best.} and Normal rewards also has regret $O(\sqrt{N})$, and that this bound is tight. Thus, the regret of a test \& roll with the profit-maximizing sample size has the same order as a multi-armed bandit with Thompson sampling. Because the actual regret depends on parameter values and can not be computed in closed-form for Thompson sampling with 
Normal response, we compare them in specific application settings in Section \ref{sec:applications}, and show that test \& roll achieves comparable average regret in several realistic cases. 

The normal model developed in this section can also be used in situations where the response is Bernoulli (e.g., clicks, purchase incidence) using the standard approximation $s=\mu(1-\mu)$ and has a convenient closed-form solution. Alternatively, Appendix \ref{sec:beta-binomial} develops a beta-binomial version where sample size must be computed numerically. Figure \ref{fig:nn_v_bb} compares exact sample sizes from the beta-binomial with the normal approximation and shows that the normal approximation provides accurate sample sizes when $\mu$ is between 0.05 and 0.95; for smaller or larger $\mu$ the sample size computed using the normal approximation is too small and we suggest using the beta-binomial formulation.

\section{Test \& Roll with Asymmetric Normal Priors}
\label{sec:asymmetric}

The analysis thus far focused on cases with a common prior for both treatments. However, there are many situations where the priors might be different, e.g., comparing a marketing communication against a holdout control. 

Relaxing the assumptions from the previous section, assume $Y_1 \sim \mathcal{N}(m_1, s_1^2)$ and $Y_2 \sim \mathcal{N}(m_2, s_2^2)$ with priors $m_1 \sim \mathcal{N}(\mu_1, \sigma_1^2)$ and $m_2 \sim \mathcal{N}(\mu_2, \sigma_2^2)$ that represent the information about the treatments available prior to the test.

Under these priors, the \emph{a priori} expected profit in the test stage is: 
\begin{align}
E[\Pi_\textnormal{T}] = \mu_1n_1 + \mu_2n_2
\label{eq:NNasym_exp_profit_test}
\end{align}

Decision rule \eqref{eq:decision_rule_post_mean} is still optimal in this case, but does not imply selecting the treatment that performs better in the test anymore; the prior information now also affects the decision. Using the decision rule in \eqref{eq:decision_rule_post_mean}, the \emph{a priori} expected profit in the roll stage is (Proposition \ref{prop:exp-prof}):
\begin{equation}
\begin{gathered}
E[\Pi_{\textnormal{D}}] = (N-n_1-n_2) \left[\mu_1 + e\Phi\left(\frac{e}{v}\right) + v\phi\left(\frac{e}{v}\right)\right]\\
\textnormal{where } e = (\mu_2 - \mu_1) \textnormal{ and } v =  \sqrt{\frac{\sigma_1^4}{\sigma_1^2+s_1^2/n_1}+\frac{\sigma_2^4}{\sigma_2^2+s_2^2/n_2}} \\
\end{gathered}
\label{eq:NNasym_exp_profit_deploy}
\end{equation}
 
% Alternatively, if the prior information is ignored and the treatment that performs better in the test is deployed, then the expected profit is:

% \begin{equation}
% \begin{aligned}
% E[\Pi_{\textnormal{D}}] = (N-n_1-n_2)\left(\mu_1 + \frac{\sigma_1^2+ \sigma_2^2}{v} \phi\left(\frac{e}{v}\right) + e \Phi\left(\frac{e}{v}\right)\right) \\
% \textnormal{where } e = (\mu_2 - \mu_1) \textnormal{ and } v = \sqrt{s_1^2/n_1+ s_2^2/n_2 + \sigma_1^2+\sigma_2^2}\\
% \end{aligned}
% \label{eq:NNasym_exp_profit_deploy_2}
% \end{equation}

The expected total profit $E[\Pi] = E[\Pi_\textnormal{T}] + E[\Pi_\textnormal{D}]$ can be maximized over $n_1$ and $n_2$ to find the optimal sample size. The optimal sample sizes can not be solved for analytically, but the function can be easily optimized numerically.\footnote{Functions for finding optimal sample sizes for asymmetric normal priors or beta-binomial priors will be included in an R package to be published to CRAN.} 

\subsection{Incumbent/Challenger Tests}
One example of an asymmetric test \& roll experiment arises when the experimenter has more past experience with treatment $1$ vs. treatment $2$, implying that $\sigma_1<\sigma_2$. We dub this an ``incumbent/challenger" test.  For example, an incumbent can be an ad copy or page design that follows the traditional firm branding strategy, while a challenger uses a new creative approach. When $\sigma_1 < \sigma_2$, the optimal sample size will be larger for the challenger treatment, to gain more information about the challenger in the test. A proof of this alongside sample size formulas can be found in Appendix \ref{sec:asymmetric_test_details}.  

\subsection{Pricing Tests}
A second common case for asymmetric test plans are pricing experiments. Because companies face uncertainty about which prices are optimal, they often experiment with multiple prices. Different prices, however, influence two important factors. First is the amount of people who will purchase the product; higher prices will elicit fewer purchases. Second is the profit per person; higher prices yield higher profits conditional on purchase. Thus, setting different prices effectively changes the priors on the mean profit per customer, which implies different optimal sample sizes for the two price levels. 

An example application that fits our framework is as following. Suppose the firm would like to pick between two known prices, $p_1$ and $p_2$, and that demand from customer $i$ presented with price $j$ is $d_{ij}=a-m\cdot p_j+\varepsilon_{ij}$. In this model, demand is linear in price, $a$ is the  willingness to pay for the product, $m$ is the uncertain price sensitivity with a prior distribution $\mathcal{N}(\mu, \sigma^2)$, and $\varepsilon_{ij} \sim \mathcal{N}(0,s^2)$. The profit from a customer $i$ presented with price $j$ will be $y_{ij}=p_j d_{ij}$. 
This model translates directly to a Normal-Normal model with asymmetric priors, when we denote $\mu_j=p_j(a-\mu p_j)$, $\sigma_j=p_j^2 \sigma$ and $s_j=p_j s$. Consequently, the profit and optimal sample size formulas derived for the asymmetric case can be applied directly to pricing experiments, and will recommend different sample sizes depending on the levels of prices being tested. A marketer could further optimize the test prices, $p_1$ and $p_2$. More distant prices help to identify $m$, but increase the opportunity cost of the test.

A more comprehensive approach to this problem was taken by \citet{misra2019}, where the goal is not to test specific prices, but rather to learn the demand curve while maximizing profits. The test \& roll setup can be adapted to solve a similar problem, but the solution will require a numerical approach for calculating sample sizes and optimal prices.

\section{Applications}
\label{sec:applications}

Designing a profit-maximizing test \& roll requires priors on the distribution of the mean response rate (profit) of the treatments. This section illustrates how to estimate these priors using data on past marketing interventions.\footnote{This is similar to using a pre-test to inform priors for conjoint design \citep{AroraHuber2001}.} We then use the estimated priors to provide optimal test plans for three different marketing contexts and compare them to hypothesis testing and multi-armed bandits using Thompson sampling, based on expected profit and regret. The first two applications use symmetric priors, while the third presents a situation where asymmetric priors are appropriate.

\subsection{Website testing}

To set priors based on past data, we analyze 2,101 website tests from \cite{berman2018} which were conducted across a wide variety of websites. For each treatment arm in each experiment we observe the click rate, $\bar{y}$ and sample size $n$.\footnote{While it would be ideal to observe sales and revenue for each visitor, this is not always possible. As a proxy, we assume for this example that profit is proportional to the number of clicks.} Fitting a hierarchical model to this data, we estimate that the mean responses (click rates) are distributed $\mathcal{N}(0.68, 0.03)$ across treatment arms. (Appendix \ref{sec:website-test-data} details the data and estimation.) 

To plan a new test, we assume this as a symmetric prior on mean response ($m_1$ and $m_2$). Assuming symmetric priors is reasonable as there is typically no prior information that one version of a web page will perform better than the other. The implied prior on the treatment effect is shown in Figure \ref{fig:website_prior} and has mean $E[|m_1 - m_2|] =  0.023$.
\begin{figure}
\centering
\includegraphics[page=1, width=0.5\textwidth]{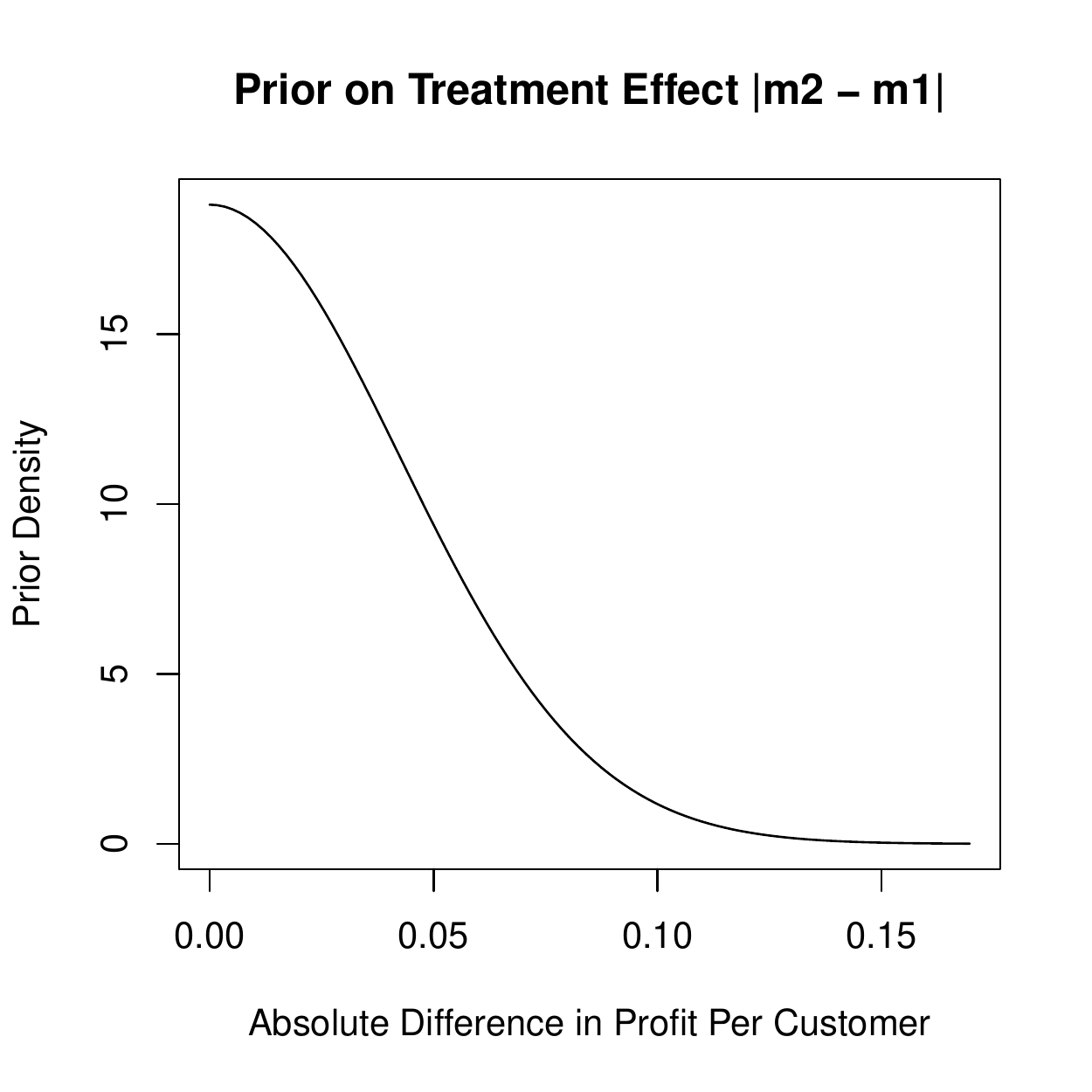}
\caption{Implied prior on treatment effect for website example.}
\label{fig:website_prior}
\end{figure}
We compute the sample size based on (\ref{eq:NN_sample_size}), using $\sqrt{\widehat{\mu}(1-\widehat{\mu})}$ to approximate $s$. The population size $N$ is set based on the expected number of people who will visit the website over the deployment period.  As an example, with $N=100,000$, the optimal test size is $n_1^* = n_2^* = 2,284$ in each test group. The expected number of clicks is 3,106 in the test and 66,430 more when the better-performing treatment is deployed, for a total of 69,536 conversions.  Following (\ref{eq:NN_error_rate}), this test will deploy the worse-performing web page 10.0\% of the time, and this represents the optimal trade-off with the opportunity cost of the test.  The profit-maximizing test \& roll has expected regret of 0.22\% relative to expected profit with perfect information\footnote{To facilitate comparisons across applications, we report the regret relative to the expected value of perfect information, i.e. the ratio of \eqref{eq:exp-regret} to \eqref{eq:exp_profit_perfect_information}.} and achieves 90.7\% of the potential gains over choosing randomly. 

Figure \ref{fig:website_n} shows the overall expected conversion rate (in the test and roll phases combined) as a function of the test size. Since small tests rapidly improve the deployment decision and increase profits, practitioners should be encouraged to run small tests and act on them. Tests that are larger than optimal decrease the error rate marginally (Figure \ref{fig:website_n}b), but erode overall expected profit (Figure \ref{fig:website_n}a). Notice that the slope of expected profit falls more swiftly when sample sizes are sub-optimal; a test is that is too large is preferable to one that is too small by the same amount. 
\begin{figure}[ht]
\centering
\begin{tabular}{cc}
(a) & (b)\\[-7ex]
\includegraphics[page=2, width=0.45\textwidth]{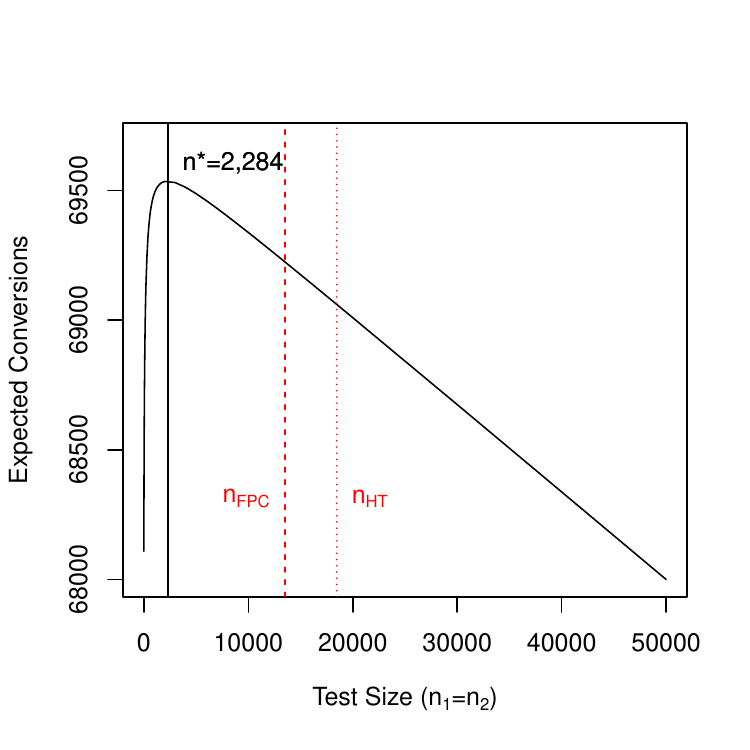} &
\includegraphics[page=3, width=0.45\textwidth]{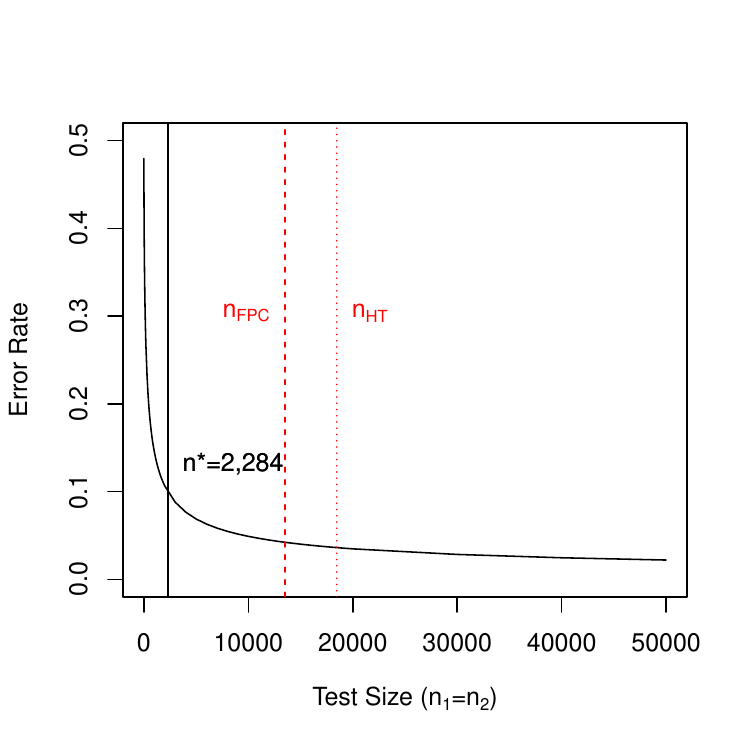}\\
\end{tabular}
\caption{Total expected profit (a) and error rate (b) as a function of test size for website test example. ($N=100,000$, $\mu=0.68$, $\sigma=0.03$, $s=0.466$).}
\label{fig:website_n}
\end{figure}

Computing the profit-maximizing test size formula in (\ref{eq:NN_sample_size}) requires the user to specify a full prior distribution on the mean response for each arm, which requires the test designer to think about how the treatments will perform and can be informed by past data. In contrast, finding the recommended sample size for a hypothesis test following \eqref{eq:nht_sample_size} requires selecting the minimum effect size to detect ($d$) and acceptable levels of type I and type II errors ($\alpha$, $\beta$). This can be challenging; many test planners have difficulty defining type I and type II error, let alone estimating the costs of those two errors to set desired levels of $\alpha$ and $\beta$. There are numerous blog posts devoted to explaining how to apply hypothesis testing to A/B tests \citep{Gershoff2017, Wortham2018}. In most applications, standard values of $\alpha=0.05$ and $\beta=0.8$ are used despite the fact that type I error is often inconsequential. 

To estimate a typical recommended sample size for a hypothesis test for this example, we use standard values for $\alpha$ and $\beta$ and set $d = 0.68 \times 0.02 = 0.0136$, i.e., a 2\% lift. This value for $d$ is the 25.1 \%-tile of the prior distribution of treatment effects implied by $\mu$ and $\sigma$. The resulting recommended sample size for a hypothesis test is 18,468 in each group (or 13,487 with a finite population correction), an order of magnitude larger than the profit-maximizing test size. This larger sample size is set to control type I and type II error tightly irrespective of the opportunity cost of the test, resulting in much larger sample sizes than are necessary to maximize expected profit. In this application, the oversized test reduces the remaining population that can receive the better treatment and results in 476 fewer expected conversions (see Figure \ref{fig:website_n} and Table \ref{tab:website_test}). 

\linespread{1.1}
\begin{table}[ht]
\caption{Comparison of test plans for website test example.
($N=100,000$, $\mu=0.68$, $\sigma=0.03$, $s=0.466$).}
\label{tab:website_test}
\begin{center}
\begin{tabular}{cccccccc}
\hline
& & & \multicolumn{3}{c}{Expected Conversions} & Exp. & Roll\\
\cline{4-6}
& $n_1$ & $n_2$ & Test & Roll & Overall & Regret & Error  \\
\hline
No Test (Random) & - & - & - & - & 68,000 & 2.43\% & 50.0\%\\
Std. Hyp. Test & 18,468 & 18,468 & 25,116 & 43,944 & 69,060 & 0.91\% & 3.6\%\\ 
Hyp. Test FPC* & 13,487 & 13,487 & 18,342 & 50,883 & 69,225 & 0.67\% & 4.2\% \\ 
Test \& Roll & 2,284 & 2,284 & 3,106 & 66,430 & 69,536 & 0.22\% & 10.0\% \\ 
Thompson Sampling & - & - & - & - & 69,672 & 0.03\% & - \\
Perfect Information & - & - & - & - & 69,693 & 0\% & - \\
\hline
\end{tabular}
\end{center}
* Hypothesis test with finite population correction
\end{table}
\linespread{1.5}

Figures \ref{fig:website_N_sigma}(a)-(c) show the sensitivity of the profit-maximizing sample size to $N$, $\sigma$ and $s$. Panel (a) shows how the sample size scales with the population $N$, allowing marketers with lower-traffic websites or pages to appropriately size website A/B tests. 
\begin{figure}
\centering
\begin{tabular}{cc}
(a) & (b) \\[-5ex]
\includegraphics[page=4, width=0.45\textwidth]{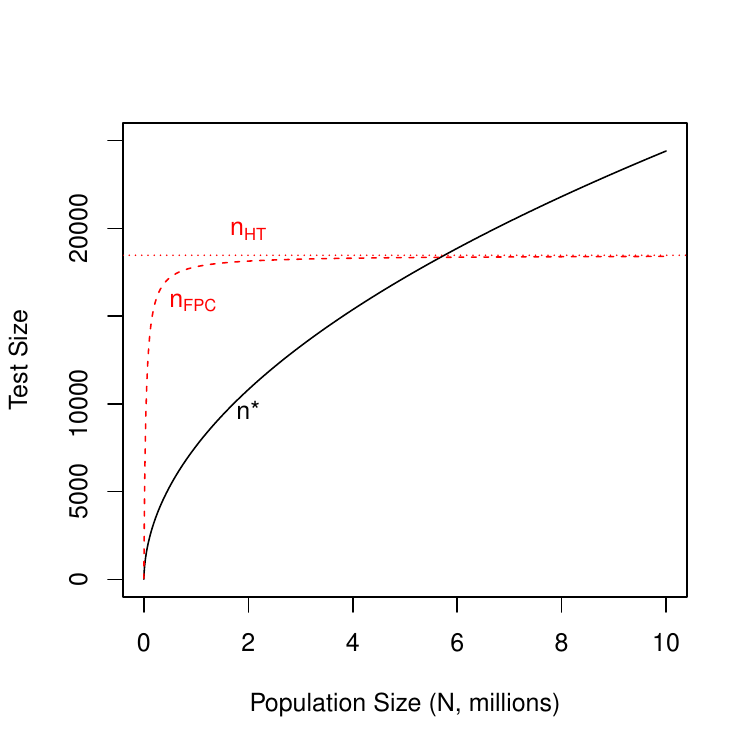} &
\includegraphics[page=6, width=0.45\textwidth]{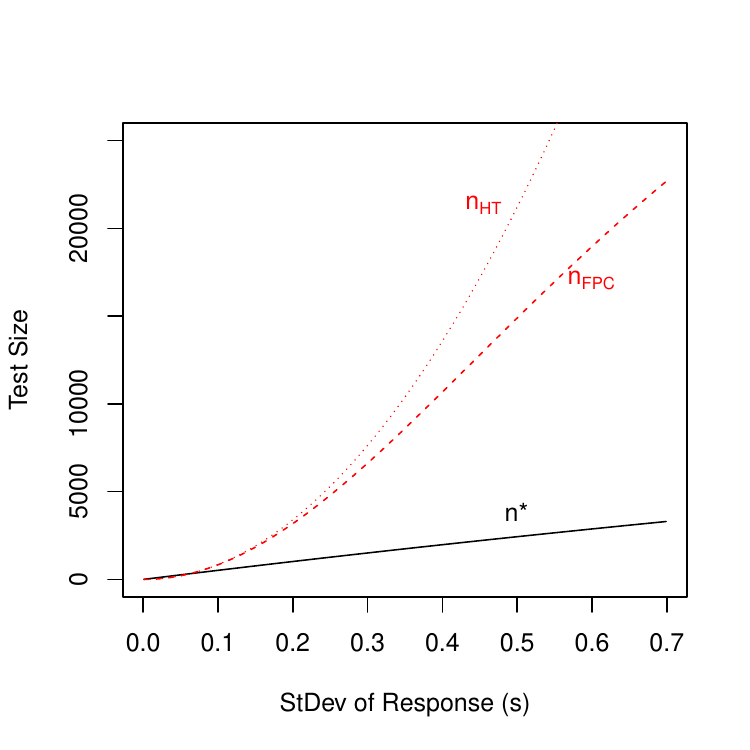}\\
(c) \\
\includegraphics[page=5, width=0.45\textwidth]{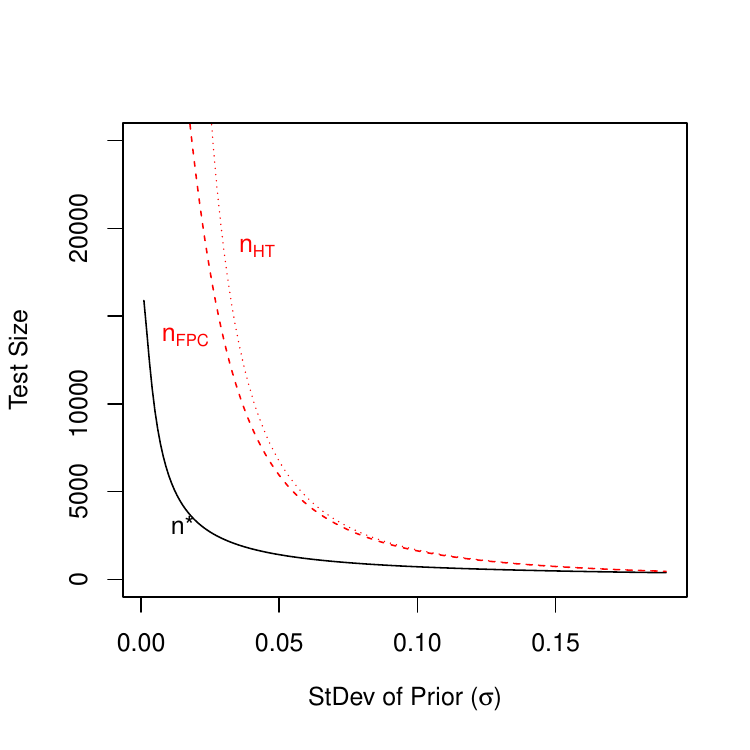}\\
\end{tabular}
\caption{Optimal sample size ($n^*$) for website test as a function of (a) Population Size $N$, (b) standard deviation of response $s$ and (c) $\sigma$ as compared to the sample size for a hypothesis test ($n_{HT}$) and a hypothesis test with finite population correction ($n_{FPC}$). Other parameters fixed at $N=100,000$, $\mu=0.68$, $\sigma=0.03$, $s=0.466$.}
\label{fig:website_N_sigma}
\end{figure}
Panel (b) shows how sample size grows linearly with the response noise $s$, unlike the recommended sample size for a null hypothesis test which increases with $s^2$. Panel (c) shows that when $\sigma$ is larger, smaller test sizes are sufficient to detect treatments that on average perform substantially better.\footnote{The values of $n_{HT}$ and $n_{FPC}$ shown in Panel (c) assume $d$ is set at the 25\%-tile of the prior of the absolute treatment effect.}

To compare test \& roll to a multi-armed bandit, Table \ref{tab:website_test} shows the expected conversions and relative regret for multi-armed bandit with Thompson sampling where units are allocated to treatments sequentially based on the posterior predictive probabilities that each treatment is best \citep{Thompson1933}. See Appendix \ref{sec:thompson_sampling_details} for implementation details. The dynamic Thompson sampling algorithm produces 136 more conversions than a test \& roll with optimal sample size. Both methods use a decision rule based on the same posterior, but the multi-armed bandit has more flexibility to recover from early observations that favor the wrong treatment. However, the difference is small: Thompson sampling achieves expected relative regret of 0.03\%, while test \& roll achieves 0.22\%. For this example, profit-maximizing test \& roll becomes an attractive option, once the operational complexity of integrating a dynamic algorithm into the website is considered. 
% \begin{figure}
% \centering
% \begin{tabular}{cc}
% (a) & (b)\\
% \includegraphics[page=1,width=0.45\textwidth,clip]{Images/website_ts.pdf} &
% \includegraphics[page=2,width=0.45\textwidth,clip]{Images/website_ts.pdf}\\
% \end{tabular}
% \caption{Distribution of relative regret for (a) profit-maximizing test \& roll and (b) multi-armed bandit with Thompson sampling for the website test example. Histograms plot 10,000 simulated experiments. ($N=100,000$, $\mu=0.68$, $\sigma=0.03$, $s=0.466$)}
% \label{fig:mab-regret}
% \end{figure}

To provide guidance as to when a test \& roll and Thompson sampling are most comparable, we compute relative regret for both algorithms, under a variety of conditions. For each condition, we simulated $R=10,000$ sets of potential outcomes on which to compare algorithms. The resulting densities of relative regret are plotted in Figure \ref{fig:regret_sensitivity}. In general, an optimized test \& roll has a wider distribution of regret with a longer right tail due to occasional deployment errors. Thompson sampling can recover from these errors and so achieves a tighter distribution of regret. The difference between algorithms is more pronounced when there are a greater number of treatment arms, where dynamic allocation provides a stronger advantage. 
\begin{figure}
\centering
\includegraphics[page=1, width=0.42\textwidth]{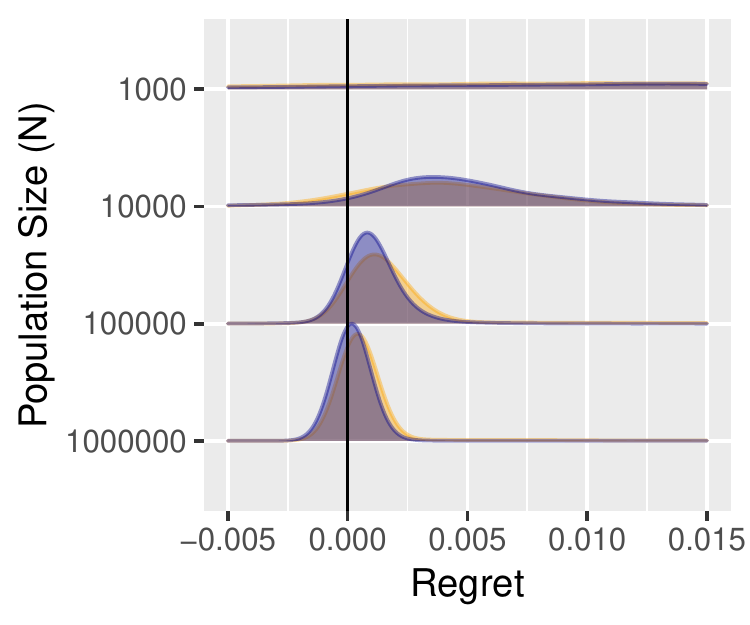}
\includegraphics[page=2, width=0.42\textwidth]{Images/website_sensitivity_ridges.pdf}\\
\includegraphics[page=3, width=0.42\textwidth]{Images/website_sensitivity_ridges.pdf}
\includegraphics[page=4, width=0.42\textwidth]{Images/website_sensitivity_ridges.pdf}\\
\includegraphics[page=5, width=0.42\textwidth]{Images/website_sensitivity_ridges.pdf}
\includegraphics[page=6, width=0.42\textwidth]{Images/website_sensitivity_ridges.pdf}\\
\caption{\linespread{1.1} \small The \emph{a priori} relative regret of Thompson sampling (dark blue) and profit-maximizing test \& roll (light orange) are remarkably similar under wide range of conditions. Parameters not varied are fixed at the website example $N=100,000$, $K=2$, $\mu=0.63$, $s=0.466$, $\sigma=0.03$. Density plots are computed from $R=10,000$ draws of potential outcomes. For $K>2$ treatments, we computed the Test \& Roll profit numerically for all possible sample sizes to find the optimum. Sometimes, the algorithm achieves profit higher than the ex-ante expected value of perfect information resulting in negative relative regret.}
\label{fig:regret_sensitivity}
\end{figure}
As discussed in Section \ref{sec:normal}, the difference is also more pronounced when there is a moderate expected difference between treatments (governed by $\sigma$), which leads to a greater risk deployment error for test \& roll. When $\sigma$ is small, there is little to be gained by selecting the right treatment. When $\sigma$ is large, the difference between treatments is large and both algorithms will detect the better treatment. Thompson sampling performs remarkably well over a wide range of conditions, usually producing relative regret less than 1\%. However, even in the worst conditions we test, the test \& roll has expected relative regret that is close to Thompson sampling, making it a reasonable alternative when there are high costs of implementing a dynamic algorithm or when greater transparency is desired.  Decision makers can compute expected regret for Thompson sampling versus test \& roll for their specific priors to evaluate whether the difference in performance exceeds the additional cost of implementing a dynamic algorithm.

\subsection{Display advertising testing}

As a second example of a profit-maximizing test \& roll, we base priors on online display ad experiments reported by \citet{LewisRao2015}. We focus on 5 experiments reported for ``Advertiser 1''. \citet{LewisRao2015} report the mean and standard deviation of the sales response (\$) in the control group for each experiment ($m_1$ and $s = s_1 = s_2$ in our notation). Applying a hierarchical model to the reported summaries, we estimate $m_1 \sim \mathcal{N}(10.36, 4.40)$ and the standard deviation of response $s$ is 103.77. See Appendix \ref{sec:estimating_priors} for details.

Ideally, we would estimate a similar distribution for the treated group, creating asymmetric priors, but \citet{LewisRao2015} do not report the treatment effects for these experiments. Instead, we assume the profit per customer $m_2$ has the same prior distribution as $m_1$. That is, on average the ads produce a lift that precisely covers the cost.   
\linespread{1.1}
\begin{table}
\footnotesize
\caption{Comparison of test plans for online display example. ($N=1,000,000$, $\mu=10.36$, $\sigma=4.40$, $s=103.77$).}
\label{tab:display_test}
\begin{center}
\begin{tabular}{cccccccc}
\hline
& & & \multicolumn{3}{c}{Expected Sales (\$000)} & & \\
\cline{4-6}
& $n_1$ & $n_2$ & Test & Roll & Overall & Regret & Roll Error  \\
\hline
No Test (Random) & - & - & - & - & 10,360 & 19.32\% & 50.0\% \\
Standard Hyp. Test* & 4,782,433* & 4,782,433* & n/a & n/a & n/a &  n/a & n/a \\
Hyp. Test FPC** & 452,673 & 452,673 & 9,380 & 1,125 & 10,595 & 17.5\% & 1.1\% \\
Test \& Roll & 11,391 & 11,391 & 236 & 12,491 & 12,727 & 0.89\% & 6.9\% \\ 
Thompson Sampling & - & - & - & - & 12,803 & 0.29\% & - \\
Perfect Information & - & - & - & - & 12,840 & 0\% & - \\
\hline
\end{tabular}
\end{center}
* Recommended test size is larger than assumed population.\\
** Hypothesis test with finite population correction as defined in Footnote \eqref{eq:nht_fpc_sample_size}.\\ 
\end{table}
\linespread{1.5}

Assuming a total population size of $N=1,000,000$, the profit-maximizing sample size is $n_1 = n_2 = 11,391$. Even with this small test, the decision of whether or not to advertise to the remainder of the population is incorrect only $6.9\%$ of the time. By contrast, these tests would require a sample size of 4,782,433 in each group for a standard hypothesis test to detect a difference of $d=0.19$ at $\alpha=0.05$ and $\beta=0.80$.\footnote{The difference of 0.19 is approximately the difference between ROI= -100\% and 0\% assuming the ads cost 0.094 per user (the average reported cost across experiments) and the margin on retail sales is 0.5. This sample size is similar to those calculated by \citet{LewisRao2015} in Table III.} As \citet{LewisRao2015} point out, tests of this size are infeasible within the budget of most advertisers and the population available on most ad platforms. Even with a finite population correction, the sample size for a hypothesis test is 452,673, which results in substantially higher regret. A risk-neutral firm can reliably determine whether advertising is more profitable than not and maximize expected profits with far smaller tests.  As can be seen by comparing \eqref{eq:nht_sample_size} and \eqref{eq:NN_sample_size}, the difference in sample size is larger when $s$ large, as it is for the display advertising tests. Even if we cut the prior variance $\sigma$ in half and increase the population to $N=10,000,000$, the profit-maximizing sample size only increases to $n_1 = n_2 = 234,361$, still half that required for a hypothesis test with finite population correction. Test sizes, profits and error rates are summarized in Table \ref{tab:display_test}. 

\subsection{Catalog holdout testing}

Finally, we illustrate how asymmetric priors described in Section \ref{sec:asymmetric} lead to unequal test group sizes. We estimate priors based on 30 catalog holdout tests conducted by a specialty retailer. For each customer in each test, we observe all-channel sales (\$) in the month after the catalog is sent. Appendix \ref{sec:estimating_priors} details how the data is used to estimate the distribution of mean catalog responses for the treated and holdout groups. Figure \ref{fig:catalog_dist} shows the fitted priors for mean revenue per customer, which are $\mathcal{N}(30.06, 13.48)$ for the treated groups and $\mathcal{N}(19.39, 20.97)$ for the holdout groups. That is, we expect the customers who receive the catalog to purchase more. The standard deviation in response within a group is estimated at $s_1 = 87.69$ and $s_2 = 179.36$. 

\begin{figure}[ht]
\centering
\includegraphics[page=1,width=0.45\textwidth]{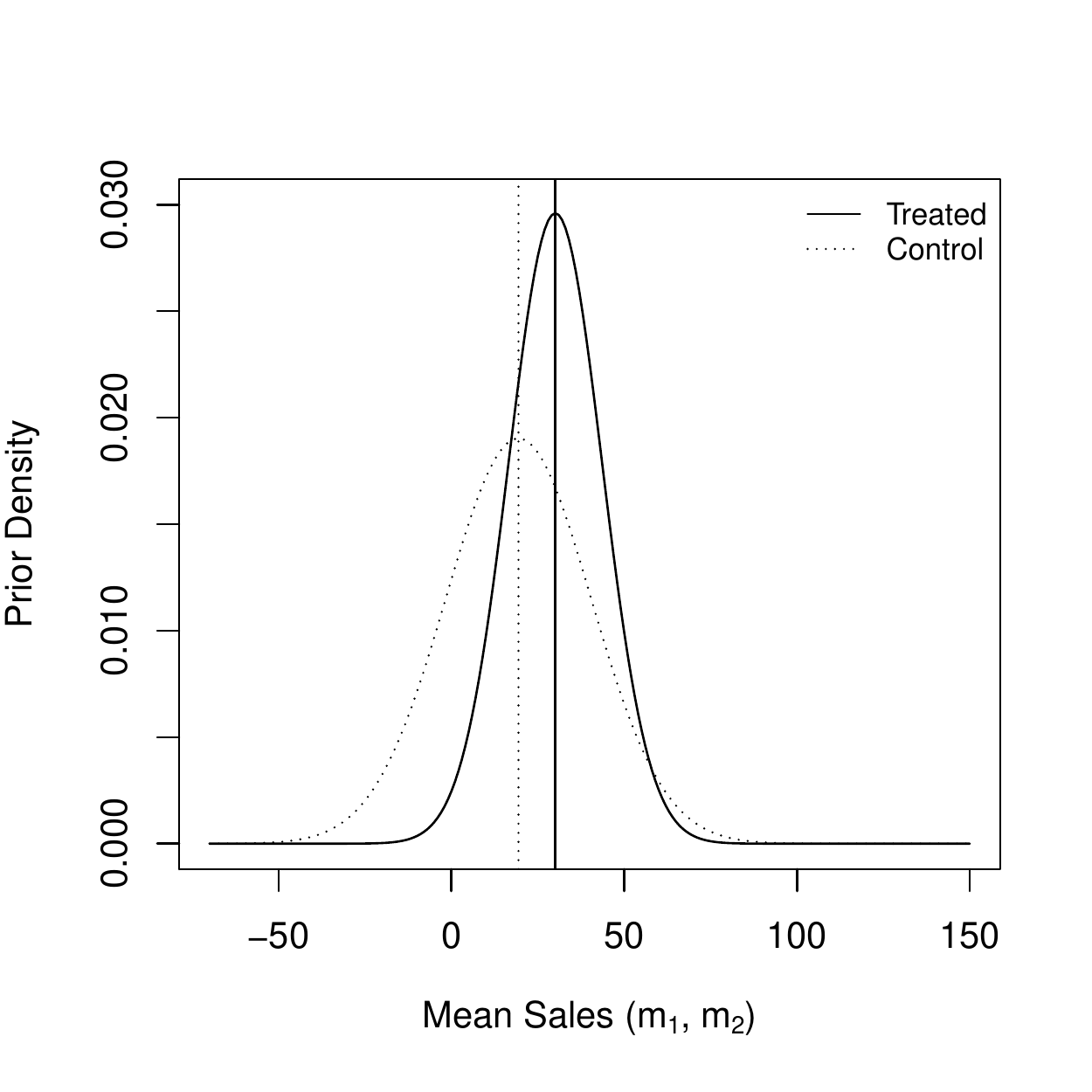}
\includegraphics[page=2,width=0.45\textwidth]{Images/catalog_priors.pdf}
\caption{Fitted distributions for mean response estimated from previous catalog mailings for a specialty retailer (left) and the implied prior on treatment effects (right).}
\label{fig:catalog_dist}
\end{figure}

After accounting for the cost of the media (approx. \$0.80), 23.2\% of catalog campaigns are expected to decrease profit based on the priors in Figure \ref{fig:catalog_dist}. A test \& roll experiment can be used with future campaigns to prevent mailing to the entire list when it is unprofitable. Assuming a population size of $N=100,000$, the profit-maximizing sample sizes are $n_1^*=588$ (control) and $n_2^*=1,884$ (treated). An experiment with these sample sizes achieves expected total sales of \$3,463,250. The recommended sample size for a hypothesis test to detect a 25\% sales lift is 7,822 in the control group and 15,996 in the treated,\footnote{When $\sigma_1 \ne \sigma_2$, then the sample sizes $n_1 = \left(z_{(1-\alpha)/2} + z_\beta \right)^2 \left(\frac{s_1^2 + s_1s_2}{\delta^2}\right)$ and $n_2 = \left(z_{(1-\alpha)/2} + z_\beta \right)^2 \left(\frac{s_1 s_2 + s_2^2}{\delta^2}\right)$ minimize $n_1 + n_2$ while achieving the desired confidence and power. See \citet{LuoGuo2007}.} resulting in a much larger test that achieves a lower expected profit of \$3,287,412. Correcting for finite sampling reduces sample sizes to 6,317 (control) and 12,921 (treated) and improves overall profit slightly. These test plans are summarized in Table \ref{tab:catalog_test}.

\begin{table}
\linespread{1.1}
\footnotesize
\caption{Comparison of test plans for catalog holdout example ($N=100,000$, $\mu_1=30.06$, $\sigma_1=13.48$, $s_1=87.69$, $\mu_2=19.39$, $\sigma_2=20.97$, $s_2=179.36$).}
\label{tab:catalog_test}
\centering
\begin{tabular}{cccccccc}
\hline
& & & \multicolumn{3}{c}{Expected Sales (\$000)} & &  \\
\cline{4-6}
& $n_1$ & $n_2$ & Test & Roll & Overall & Regret & Roll Error  \\
\hline
No Test (Random) & - & - & - & - & 2,433 & 30.75\% & 50.0\% \\
Hypothesis Test & 7,822 & 15,999 & 620 & 2,668 & 3,287 & 7.85\% & 1.9\%\\ 
Hypothesis Test (FPC) & 6,317 & 12,921 & 501 & 2,828 & 3,328 & 5.23\%& 2.2\%\\
Test \& Roll & 588 & 1,884 & 67 & 3,409 & 3,476 & 1.68\% & 6.4\%\\
Thompson Sampling & - & - & - & - & 3,504 & 0.57\% & - \\
Perfect Information & - & - & - & - & 3,512 & 0\% & - \\
\hline
\end{tabular}
\end{table}
%\linespread{1.5}

The profit-maximizing test and the null hypothesis test both allocate a larger sample to the treatment group, but for different reasons. The hypothesis test does so because the treatment group has a noisier response ($s_1 < s_2$). The profit-maximizing test additionally considers that we \emph{a priori} expect greater profits from customers who receive the catalog ($m_1 < m_2$).  Even if we fix $s_1 = s_2$ and re-estimate the hierarchical model (see Appendix \ref{sec:estimating_priors}), the resulting test \& roll sample size is $n_1 = 771$ and $n_2 = 1,949$, due to the remaining differences in the priors. 

Figure \ref{fig:catalog_sample_size_lift} shows the sensitivity of the sample sizes to the expected catalog lift.  We analyzed this sensitivity by varying $\mu_2$, leaving all other parameters of the priors fixed. As the plot shows, when the expected lift is very high, a small holdout group is optimal. Thus, the common practice of using small holdout tests can be rationalized by a prior expectation that the treatment increases sales (or other desired behavior) more than the cost of marketing. The test \& roll framework provides a principled way to set the size of the holdout group by making these priors explicit.
\begin{figure}[!ht]
\centering
\includegraphics[width=0.45\textwidth]{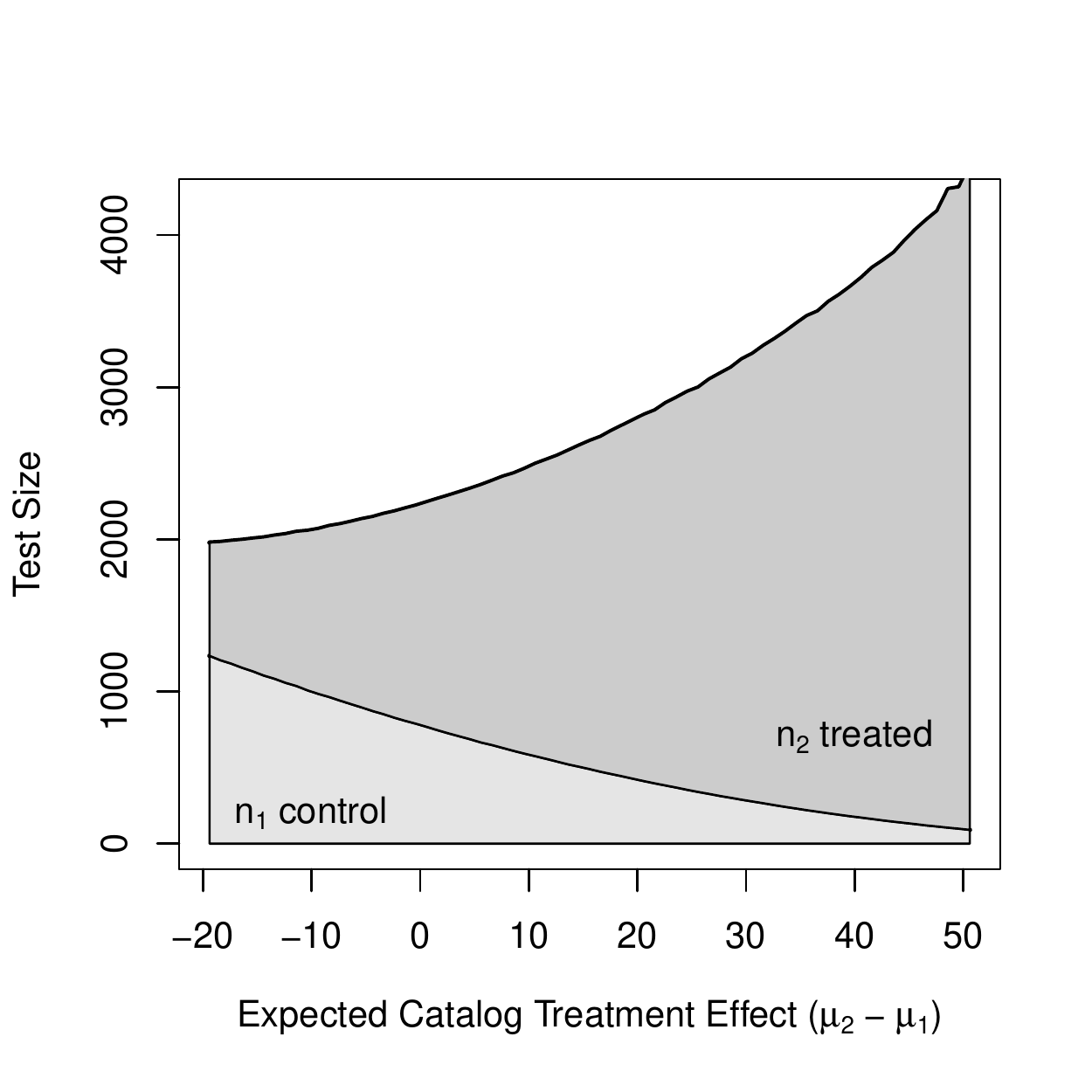}
\caption{Sensitivity of optimal test sizes to the \emph{a priori} expected increase in sales from the catalog ($\mu_2 - \mu_1$).}
\label{fig:catalog_sample_size_lift}
\end{figure}

\section{Discussion}
\label{sec:discussion}
We present a new approach to planning sample sizes for A/B tests. Unlike the classic hypothesis test that emphasizes high confidence and power, our approach optimally balances the trade-off between not deploying the best treatment in the roll stage and the cost of identifying this treatment in the test stage. The practical result is far smaller recommended test sizes that scale to the size of the available population. Most importantly, by focusing on profit, we show that marketers should not be discouraged from running smaller tests and acting on the findings; while imperfect, such smaller tests increase profit. Profit-maximizing tests may split the test sample unequally between the treatments allowing us to rationalize this common practice in marketing experiments.

The profit-maximizing sample size is optimized for marketing campaigns, which typically have a limited target population. Direct marketing campaigns are conducted with finite mailing lists. Media campaigns have a fixed budget. Webpages have limited traffic. With finite populations, the firm should identify which treatment to deploy to the majority of the population without ``wasting'' too many exposures on suboptimal treatments in the test. 

Unlike fully-dynamic approaches \citep{bertsimas2007, chick2012, schwartzBradlowFader2016WP} that vary allocation continuously, our method fits within the typical A/B testing framework, and requires no changes in testing software other than the recommended sample size. Operational complexity is reduced by providing a definitive end to the test phase, limiting the number of alternative treatments that must be maintained and providing transparency about what treatment is being selected, what evidence led to the selection of this treatment and what the expected benefit (or regret) is.  Managers can interject if they wish before ``rolling." These features make the profit-maximizing test \& roll attractive to marketers. 

One limitation of our method is that the best treatment will not always be selected. Although the error-rate may be higher than the one guaranteed by typical null hypothesis testing, the profit-maximizing test size sets the error rate optimally, based on the potential differences between treatments and resulting opportunity costs. In contexts where the decision maker is risk averse or the cost of deploying a subpar treatment is very high, as in clinical trials \citep{berry1994,cheng2003}, then other approaches are warranted.

Further extensions of the test \& roll framework presented in Section \ref{sec:model} would be useful. As data from sets of experiments becomes available \citep{johnson2017, bartstephensarvary2014}, there is opportunity to develop a catalog of priors for different test settings. Other forms for prior distributions could be considered. For example, \citet{stallard2017} extend the test \& roll framework to response distributions from the exponential family using approximations. \citep{Azevedo2019} focus on priors with fat tails.

The test \& roll method is easily extended to more than two treatments, potentially allowing for correlated priors, e.g., for a holdout group versus several alternative marketing treatments.  The cost of switching between treatments, which can be substantial for offline marketing treatments, could also be incorporated into the decision problem. If it is possible to deploy different treatments to sub-populations, then the potential to identify heterogeneous treatment effects \citep{hitsch2018, simester2018} can be considered in the test design. Similarly, time dependency in response could be considered, e.g. day-of-week or ``novelty'' effects.  These extensions all fit naturally within the test \& roll framework.

\bibliography{references}

\renewcommand{\appendixpagename}{Appendix}

\newtheorem{prop}{Proposition} \numberwithin{prop}{section}
\newtheorem{cor}[prop]{Corollary} \numberwithin{prop}{section}

\renewcommand{\theprop}{\Alph{section}.\arabic{prop}}

\begin{appendices}

\section{Normal-Normal Model Derivations}
\label{sec:NN_derivations}

\begin{prop}[Expected roll stage profit]
\label{prop:exp-prof}
When the mean profit $y_j$ is distributed $y_j \sim \mathcal{N}(m_j,s_j^2/n_j)$ with prior $m_j \sim \mathcal{N}(\mu_j,\sigma_j^2)$, and when the decision rule picks the arm with the highest posterior mean, the expected profit in the roll stage is:
\begin{align}
& E[\Pi_D] = (N-n_1-n_2)\left[\sqrt{\frac{\sigma_1^4}{\sigma_1^2+s_1^2/n_1}+\frac{\sigma_2^4}{\sigma_2^2+s_2^2/n_2}}\phi\left(\frac{\mu_1-\mu_2}{\sqrt{\frac{\sigma_1^4}{\sigma_1^2+s_1^2/n_1}+\frac{\sigma_2^4}{\sigma_2^2+s_2^2/n_2}}}\right)\right.\nonumber\\
&\left.+(\mu_1-\mu_2)\Phi\left(\frac{\mu_1-\mu_2}{\sqrt{\frac{\sigma_1^4}{\sigma_1^2+s_1^2/n_1}+\frac{\sigma_2^4}{\sigma_2^2+s_2^2/n_2}}}\right)\right]
\end{align}
\end{prop}
\begin{proof}
Denote the decision rule $\delta(y_1,y_2) = I(a_1+b_1 y_1>a_2+b_2 y_2)$. The linear decision rule includes the optimal one that uses the posterior predictive distribution with $a_j = \frac{s_j^2/n_j \mu_j}{\sigma_j^2+s_j^2/n_j}$ and $b_j = \frac{\sigma_j^2}{\sigma_j^2+s_j^2/n_j}$.
Denote the pdf of $y_j$ as $f_j$ and its cdf as $F_j$. Denote the pdf of $m_j$ as $g_j$ and its cdf as $G_j$.

The expected value from the roll stage is:
{
\footnotesize
\begin{equation}\label{eq:exp}
E[\Pi_D]=\int_{m_1} \int_{m_2} \int_{y_1} \int_{y_2} (N-n_1-n_2)\left(\delta(y_1,y_2) m_1 + (1-\delta(y_1,y_2)) m_2 \right) f_2(y_2) f_1(y_1) g_2(m_2) g_1(m_1) dy_2 dy_1 dm_2 dm_1 
\end{equation}
}
In the derivation, we will make multiple uses of the following identities:
\begin{equation}\label{eq:yFf}
    \int_{-\infty}^{\infty} y \Phi\left(\frac{y+b}{a}\right) \phi(y) dy = \frac{1}{\sqrt{a^2+1}} \phi \left(\frac{b}{\sqrt{a^2+1}}\right)
\end{equation}
and:
\begin{equation}\label{eq:Ff}
    \int_{-\infty}^{\infty} \Phi\left(\frac{y+b}{a}\right) \phi(y) dy = \Phi \left(\frac{b}{\sqrt{a^2+1}}\right)
\end{equation}

The expression $(N-n_1-n_2)$ can be taken out of the integrand.
Continuing with the first additive in the integral (the second will be symmetric):
\begin{align}
    & \int_{y_1} \int_{y_2} \delta(y_1,y_2) m_1 f_2(y_2) f_1(y_1) dy_1 dy_2 \\
    & = \int_{y_1} \int_{-\infty}^{\frac{a_1-a_2+b_1 y_1}{b_2}} m_1 f_2(y_2) f_1(y_1) dy_1 dy_2 \\
    & = m_1 \int_{y_1} F_2\left(\frac{a_1-a_2+b_1 y_1}{b_2}\right) f_1(y_1) dy_1 \\
    & = m_1 \int_{y_1} \Phi\left(\frac{\frac{a_1-a_2+b_1 y_1}{b_2}-m_2}{s_2/\sqrt{n_2}}\right) \frac{1}{s_1/\sqrt{n_1}}\phi \left(\frac{y_1-m_1}{s_1/\sqrt{n_1}}\right) dy_1 \\
    & = m_1 \int_{y} \Phi\left(\frac{y+\frac{a_1-a_2+b_1 m_1-b_2 m_2}{b_1 s_1/\sqrt{n_1}}}{\frac{b_2 s_2/\sqrt{n_2}}{b_1 s_1/\sqrt{n_1}}}\right) \phi(y) dy
\end{align}
The last equation uses $y=\frac{y_1-m_1}{s_1/\sqrt{n_1}}$ as a change of variables.

Using identity (\ref{eq:Ff}), the final integral equals:
\begin{align}
& m_1 \int_{y} \Phi\left(\frac{y+\frac{a_1-a_2+b_1 m_1-b_2 m_2}{b_1 s_1/\sqrt{n_1}}}{\frac{b_2 s_2/\sqrt{n_2}}{b_1 s_1/\sqrt{n_1}}}\right) \phi(y) dy \\
& = m_1 \Phi\left(\frac{a_1-a_2+b_1 m_1-b_2 m_2}{\sqrt{b_1^2 s_1^2/n_1+b_2^2 s_2^2/n_2}}\right)
\end{align}

Plugging back into the expected value in (\ref{eq:exp}), the expected value of the roll stage equals:
\begin{align}
    (N-n_1-n_2) & \left(\int_{m_1} \int_{m_2} m_1 \Phi\left(\frac{a_1-a_2+b_1 m_1-b_2 m_2}{\sqrt{b_1^2 s_1^2/n_1+b_2^2 s_2^2/n_2}}\right) g_2(m_2) g_1(m_1) dm_1 dm_2 \right. \nonumber \\
    & \left. + \int_{m_2} \int_{m_1} m_2 \Phi\left(\frac{a_2-a_1+b_2 m_2-b_1 m_1}{\sqrt{b_1^2 s_1^2/n_1+b_2^2 s_2^2/n_2}}\right) g_1(m_1) g_2(m_2) dm_1 dm_2\right)
\end{align}

Using identity (\ref{eq:Ff}) again, the first additive equals:
\begin{align}
    & \int_{m_1} \int_{m_2} m_1 \Phi\left(\frac{a_1-a_2+b_1 m_1-b_2 m_2}{\sqrt{b_1^2 s_1^2/n_1+b_2^2 s_2^2/n_2}}\right) g_2(m_2) g_1(m_1) dm_1 dm_2 \\
    & = \int_{m_1} m_1 \int_{m} \left(1-\Phi\left(\frac{m+\frac{a_2-a_1-b_1 m_1+b_2 \mu_2}{b_2 \sigma_2}}{\frac{\sqrt{b_1^2 s_1^2/n_1+b_2^2 s_2^2/n_2}}{b_2 \sigma_2}}\right)\right) \phi(m) dm g_1(m_1) dm_1 \\
    & = \int_{m_1} m_1 \Phi\left(\frac{a_1-a_2 +b_1 m_1 -b_2 \mu_2}{\sqrt{b_1^2 s_1^2/n_1+b_2^2 s_2^2/n_2+b_2^2 \sigma_2^2}}\right) \frac{1}{\sigma_1} \phi\left(\frac{m_1-\mu_1}{\sigma_1}\right) dm_1 \\
    & = \int_m (m \sigma_1+\mu_1) \Phi \left(\frac{m+\frac{a_1-a_2+b_1 \mu_1 -b_2 \mu_2}{b_1 \sigma_1}}{\sqrt{\frac{b_1^2 s_1^2/n_1+b_2^2 s_2^2/n_2+b_2^2\sigma_2^2}{b_1^2 \sigma_1^2}}}\right) \phi(m) dm
\end{align}
where the last equation uses the change of variables $m=\frac{m_1-\mu_1}{\sigma_1}$.

Using identities (\ref{eq:yFf}) and (\ref{eq:Ff}), we receive:
\begin{align}
    & \int_m (m \sigma_1+\mu_1) \Phi \left(\frac{m+\frac{a_1-a_2+b_1 \mu_1 -b_2 \mu_2}{b_1 \sigma_1}}{\sqrt{\frac{b_1^2 s_1^2/n_1+b_2^2 s_2^2/n_2+b_2^2\sigma_2^2}{b_1^2 \sigma_1^2}}}\right) \phi(m) dm \\
    & = \frac{b_1 \sigma_1^2}{\sqrt{b_1^2 s_1^2/n_1+b_2^2 s_2^2/n_2+b_1^2 \sigma_1^2+b_2^2 \sigma_2^2}} \phi\left(\frac{a_1-a_2+b_1\mu_1-b_2 \mu_2}{\sqrt{b_1^2 s_1^2/n_1+b_2^2 s_2^2/n_2+b_1^2 \sigma_1^2+b_2^2 \sigma_2^2}}\right) \nonumber \\ & + \mu_1 \Phi\left(\frac{a_1-a_2+b_1\mu_1-b_2 \mu_2}{\sqrt{b_1^2 s_1^2/n_1+b_2^2 s_2^2/n_2+b_1^2 \sigma_1^2+b_2^2 \sigma_2^2}}\right)
\end{align}

Using symmetry, the \emph{a priori} expected value of the roll stage is:
{
\footnotesize
\begin{align}
E[\Pi_D] & = (N-n_1-n_2)\left[\frac{b_1 \sigma_1^2+b_2 \sigma_2^2}{\sqrt{b_1^2 s_1^2/n_1+b_2^2 s_2^2/n_2+b_1^2 \sigma_1^2+b_2^2 \sigma_2^2}} \phi\left(\frac{a_1-a_2+b_1\mu_1-b_2 \mu_2}{\sqrt{b_1^2 s_1^2/n_1+b_2^2 s_2^2/n_2+b_1^2 \sigma_1^2+b_2^2 \sigma_2^2}}\right)\right. \nonumber\\ 
& \left.+ (\mu_1-\mu_2) \Phi\left(\frac{a_1-a_2+b_1\mu_1-b_2 \mu_2}{\sqrt{b_1^2 s_1^2/n_1+b_2^2 s_2^2/n_2+b_1^2 \sigma_1^2+b_2^2 \sigma_2^2}}\right)\right]
\end{align}
}

Plugging in the posterior mean parameters for $a_j$ and $b_j$ (as they are optimal), the roll stage expected value in the fully asymmetric model is:
\begin{align}
& E[\Pi_D] = (N-n_1-n_2)\left[\sqrt{\frac{\sigma_1^4}{\sigma_1^2+s_1^2/n_1}+\frac{\sigma_2^4}{\sigma_2^2+s_2^2/n_2}}\phi\left(\frac{\mu_1-\mu_2}{\sqrt{\frac{\sigma_1^4}{\sigma_1^2+s_1^2/n_1}+\frac{\sigma_2^4}{\sigma_2^2+s_2^2/n_2}}}\right)\right.\nonumber\\
&\left.+(\mu_1-\mu_2)\Phi\left(\frac{\mu_1-\mu_2}{\sqrt{\frac{\sigma_1^4}{\sigma_1^2+s_1^2/n_1}+\frac{\sigma_2^4}{\sigma_2^2+s_2^2/n_2}}}\right)\right]
\end{align}
where in the text we set $e = \mu_1 - \mu_2$ and $v =  \sqrt{\frac{\sigma_1^4}{\sigma_1^2+s_1^2/n_1}+\frac{\sigma_2^4}{\sigma_2^2+s_2^2/n_2}}$ in Equation (\ref{eq:NNasym_exp_profit_deploy}). Thus we have completed the proof for the asymmetric case.

To get the expression in (\ref{eq:NN_exp_profit_deploy_symmetric}) we plug-in $\mu_1=\mu_2=\mu$, $\sigma_1=\sigma_2=\sigma$ and $s_1=s_2=s$ into the above expression.
\end{proof}

\begin{prop}[Profit maximizing sample size]
\label{prop:sample-size}
When the mean profits $y_j$ are distributed $y_j \sim \mathcal{N}(m_j,s^2/n_j)$ with prior $m_j \sim \mathcal{N}(\mu,\sigma^2)$, the profit-maximizing sample size is: 
\[n_1^* = n_2^* %= \frac{\sqrt{9\sigma^4 + 4 n \sigma^2 \sigma_0^2} - \frac{3}{4} \sigma^2}{4\sigma_0^2} 
= \sqrt{\frac{N}{4}\left( \frac{s}{\sigma} \right)^2 + \left( \frac{3}{4} \left( \frac{s}{\sigma} \right)^2  \right)^2 } -  \frac{3}{4} \left(\frac{s}{\sigma} \right)^2 \]
\end{prop}
\begin{proof}
Because the priors are symmetric, the optimal sample sizes will be equal.
Denote them as $n=n_1=n_2$.

The expected profit of the experiment with symmetric priors is:
\begin{equation}
E[\Pi_T]+E[\Pi_D]=N\mu+(N - 2 n) \left[\frac{\sqrt{2} \sigma^2}{\sqrt{\pi} \sqrt{2\sigma^2 + \frac{2}{n}s^2} }\right]
\end{equation}
The FOC w.r.t to $n$ is:
\begin{equation}
\frac{\sigma ^2 \sqrt{\frac{s^2}{n}+\sigma ^2} \left(4 n^2 \sigma ^2+6 n s^2-N s^2\right)}{2 \sqrt{\pi } \left(n \sigma ^2+s^2\right)^2} =0
\end{equation}
which is equivalent to solving $4 n^2 \sigma ^2+6 n s^2-N s^2=0$, yielding the optimal sample size formula.
\end{proof}

\begin{cor}[Expected error rate]
\label{cor:error-rate}
Under symmetric priors, the expected rates of making the incorrect choice in the roll stage, $E[Pr(\delta(y_1, y_2) = 1 | m_1 < m_2)]$ and $E[Pr(\delta(y_1, y_2) = 0 | m_1 > m_2)]$, both equal $\frac{1}{4} - \frac{1}{2 \pi}\arctan\left(\frac{\sqrt{2}\sigma}{s} \sqrt{\frac{n_1 n_2}{n_1+n_2}}\right)$.
\end{cor}
\begin{proof}
Using the fact that $y_j \sim \mathcal{N}(m_j,s^2/n_j)$ and because in the symmetric case the decision rule is to pick the treatment with the highest mean:
\begin{equation}
    Pr(\delta(y_1,y_2)=1|m_1,m_2) = Pr(y_1-y_2>0|m_1,m_2) = \Phi\left(\frac{m_1-m_2}{s\sqrt{\frac{1}{n_1}+\frac{1}{n_2}}}\right)
\end{equation}

If we denote $m=m_1-m_2$, then $m_1-m_2$ has a prior $\mathcal{N}(0,2\sigma^2)$. The expected error rate is therefore:
\begin{align}
    E[\delta(y_1,y_2)=1|m_1>m_2] & = \int_{-\infty}^0 Pr(y_1-y_2>0|m) Pr(m) dm \nonumber \\
    & = \int_{-\infty}^0 \Phi\left(\frac{m}{s\sqrt{\frac{1}{n_1}+\frac{1}{n_2}}}\right) \frac{1}{\sqrt{2}\sigma} \phi \left(\frac{m}{\sqrt{2}\sigma}\right)dm
\end{align}

Using the identity $\int_{-\infty}^0 \phi(ax)\Phi(bx)dx = \frac{1}{2\pi |a|}\left(\frac{\pi}{2}-\arctan\left(\frac{b}{|a|}\right)\right)$, we get the expression:
\[ E[\delta(y_1,y_2)=1|m_1>m_2] = E[\delta(y_1,y_2)=0|m_1<m_2] = \frac{1}{4} - \frac{1}{2 \pi}\arctan\left(\frac{\sqrt{2}\sigma}{s} \sqrt{\frac{n_1 n_2}{n_1+n_2}}\right)\]
\end{proof}

\begin{prop}[Regret]
\label{prop:regret}
In the symmetric Normal-Normal model with a population size $N$:
\begin{enumerate}
    \item The expected value of perfect information is $E[\Pi|\textnormal{PI}] = N\left(\mu + \frac{\sigma}{\sqrt{\pi}}\right)$.
    \item The regret of the profit-maximizing design is $O(\sqrt{N})$.
    \item The regret from using a classic hypothesis test is $\Omega(N)$.
\end{enumerate}
\end{prop}
\begin{proof}
Perfect information allows the marketer to pick the treatment with the highest mean $m_j$ without testing, yielding expected profit of $N \cdot E[max(m_1,m_2)]$.
Because both treatments come from the same prior $\mathcal{N}(\mu,\sigma^2)$, the mean of the maximum of two i.i.d Normal variables is $\mu + \frac{\sigma}{\sqrt{\pi}}$, proving the first item.

To prove the second item, we calculate the regret from using the profit maximizing design:
\begin{align}\label{prof-max-regret}
    E[\Pi|PI] - E[\Pi^*_D+\Pi^*_T] & = N \frac{\sigma}{\sqrt{\pi}}\left(1-\frac{\sigma}{\sqrt{\sigma^2+\frac{s^2}{n^*}}}\right)+\frac{2n^*\sigma^2}{\sqrt{\pi}\sqrt{\sigma^2+\frac{s^2}{n^*}}}
\end{align}

Using the inequality $\sqrt{x+1}-\sqrt{x} < \frac{1}{2 \sqrt{x}}$ for $x>0$, and denoting $x=n^*\sigma^2/s^2$, the first additive results in:
\begin{align}
    & N \frac{\sigma}{\sqrt{\pi}}\left(1-\frac{\sigma}{\sqrt{\sigma^2+\frac{s^2}{n^*}}}\right) \\ 
    & \le N \frac{\sigma}{\sqrt{\pi}} \frac{1}{2 \sqrt{\sigma^2 n^* /s^2 +1}\sqrt{\sigma^2 n^* /s^2}} \\
    & \le N \frac{\sigma}{\sqrt{\pi}} \frac{1}{2 n^*\sigma^2/s^2} \\
\end{align}

Plugging in $n^*=\sqrt{\frac{N}{4}\left( \frac{s}{\sigma} \right)^2 + \left( \frac{3}{4} \left( \frac{s}{\sigma} \right)^2  \right)^2 } -  \frac{3}{4} \left(\frac{s}{\sigma} \right)^2$, the denominator $2 n^*\sigma^2/s^2$ is larger than $\frac{1}{2}\frac{\sigma}{s}\sqrt{N}$ when $N>4 \frac{s^2}{\sigma^2}$.
Hence, we can bound the first additive in the regret (\ref{prof-max-regret}) from above by:
\begin{equation}
    N \frac{\sigma}{\sqrt{\pi}}\left(1-\frac{\sigma}{\sqrt{\sigma^2+\frac{s^2}{n^*}}}\right) \le \frac{2s\sqrt{N}}{\sqrt{\pi}}
\end{equation}

To bound the second additive:
\begin{align}
    \frac{2n^*\sigma^2}{\sqrt{\pi}\sqrt{\sigma^2+\frac{s^2}{n^*}}} \le \frac{2n^*\sigma^2}{\sqrt{\pi}\sqrt{\sigma^2}} = \frac{2n^*\sigma}{\sqrt{\pi}} < \frac{s\sqrt{N}}{\sqrt{\pi}}
\end{align}
The first inequality uses the fact that $\frac{s^2}{n^*}$ is positive, while the second uses the fact that $n^* < \sqrt{N}\frac{s}{2\sigma}$ as shown in the main text.

Summing the two additives shows that the regret of the profit maximizing design is smaller than $\frac{3s \sqrt{N}}{\sqrt{\pi}}$ proving the second item that the regret is $O(\sqrt{N})$.

To prove the third item, we plug-in the sample size from (\ref{eq:nht_sample_size}) for $n$ in the regret formula:
\begin{align}
    E[\Pi|PI] - E[\Pi_D+\Pi_T] & = N \frac{\sigma}{\sqrt{\pi}}\left(1-\frac{\sigma}{\sqrt{\sigma^2+\frac{s^2}{n}}}\right)+\frac{2n\sigma^2}{\sqrt{\pi}\sqrt{\sigma^2+\frac{s^2}{n}}} \\
    & > N \frac{\sigma}{\sqrt{\pi}}\left(1-\frac{\sigma}{\sqrt{\sigma^2+\frac{s^2}{n}}}\right) = N \frac{\sigma}{\sqrt{\pi}}\left(1-\frac{\sigma}{\sqrt{\sigma^2+\frac{\delta^2}{2z^2}}}\right) \label{nhst-eq-1}\\
    & > N \frac{\sigma}{\sqrt{\pi}}\frac{1}{2(\frac{2z^2\sigma^2}{\delta^2}+1)} = \Omega(N)\label{nhst-eq-2}
\end{align}
where the equality in (\ref{nhst-eq-1}) follows from plugging-in the NHST sample size denoting $z=z_{(1-\alpha)/2}+z_\beta$, and the last inequality follows from $\sqrt{x+1}-\sqrt{x}>\frac{1}{2\sqrt{x+1}}$ when $x \ge 0$, with $x=n \sigma^2/s^2$.

When using the sample size for the NHST with finite population correction, we can use the same approach where in Equation \eqref{nhst-eq-1} we plug-in $n = \frac{ (z_{1-\alpha/2}+z_{\beta})^2 2 s^2 N}{ (N-1) d^2 + 4 s^2 (z_{1-\alpha/2}+z_{\beta})^2}$.

This results in:
\begin{align}
    E[\Pi|PI] - E[\Pi_D+\Pi_T] & = N \frac{\sigma}{\sqrt{\pi}}\left(1-\frac{\sigma}{\sqrt{\sigma^2+\frac{s^2}{n}}}\right)+\frac{2n\sigma^2}{\sqrt{\pi}\sqrt{\sigma^2+\frac{s^2}{n}}} \\
    & > N \frac{\sigma}{\sqrt{\pi}}\left(1-\frac{\sigma}{\sqrt{\sigma^2+\frac{s^2}{n}}}\right) > N \frac{\sigma}{\sqrt{\pi}}\frac{1}{2(\frac{n\sigma^2}{s^2}+1)}\\
    & = N \frac{\sigma}{\sqrt{\pi}}\frac{1}{2\left(\frac{2 N z^2 \sigma^2}{(N-1)d^2+4z^2 s^2}+1\right)} > N \frac{\sigma}{\sqrt{\pi}}\frac{1}{2\left(\frac{2 N z^2 \sigma^2}{(N-1)d^2}+1\right)} \\
    & > N \frac{\sigma}{\sqrt{\pi}}\frac{1}{2\left(\frac{2 N z^2 \sigma^2}{1/2 N d^2}+1\right)} = N \frac{\sigma}{\sqrt{\pi}}\frac{1}{2\left(\frac{4 z^2 \sigma^2}{d^2}+1\right)} = \Omega(N)
\end{align}
\end{proof}
The last inequality follows from the fact that $1/2N \le N-1$ for $N \ge 2$, which completes the proof.

\begin{cor}[Maximum relative regret]
\label{cor:max-regret}
There is an intermediate value of $\sigma$ for which the profit-maximizing test \& roll achieves the maximum relative regret. 
\end{cor}
\begin{proof}
The relative regret at the optimal sample size equals: \[\frac{\sqrt{\frac{\sqrt{4 N \sigma ^2+9 s^2}-3 s}{\sqrt{4 N \sigma ^2+9 s^2}+s}} \left(s \left(\sqrt{4 N \sigma ^2+9 s^2}-3 s\right)+2 \sigma  \left(\sqrt{N \left(s \left(\sqrt{4 N \sigma ^2+9 s^2}+3 s\right)+N \sigma ^2\right)}-N \sigma \right)\right)}{2 N \sigma  \left(\sqrt{\pi } \mu +\sigma \right)}\]

Using L'Hôpital's rule, the limit as $\sigma \rightarrow 0$ is zero. Similarly, the limit as $\sigma \rightarrow \infty$ is zero. The relative regret is always positive. Consequently, the relative regret achieves a maximum for a value of $\sigma$ which is not zero or infinity.
\end{proof}

\section{Derivations for Beta-Binomial model}
\label{sec:beta-binomial}

 Let the profit $y_{ij}$ from customer $i$ exposed to treatment arm $j$ be $v_j$ with probability $p_j$ and $0$ with probability $1-p_j$, and let $y_j=\frac{\sum_{i=1}^{n_j} y_{ij}}{n_j}$ be the average number of conversions with treatment $j$, when $n_j$ is the number of individuals assigned to treatment $j$. We put a $Beta(\alpha,\beta)$ prior distribution on $p_j$ and denote its pdf as $f(\cdot)$.

\begin{prop}[Beta-Binomial expected profit] If profit $y_{ij}$ from customer $i$ exposed to treatment $j$ is $v_j$ with probability $p_j$ and zero otherwise with priors $p_j \sim Beta(\alpha,\beta)$:
\begin{enumerate}
\item The expected profit in the test stage is $ (n_1 v_1+n_2 v_2) \frac{\alpha}{\alpha+\beta}$
\item The expected profit in the roll stage is:
{\footnotesize
\begin{align}
    & (N-n_1-n_2)  \sum_{y_2=1}^{n_2} \binom{n_2}{y_2}\frac{\Gamma (y_2+\alpha ) \Gamma (n_2-y_2+\beta )}{B(\alpha ,\beta ) \Gamma (n_2+\alpha +\beta )} \cdot \nonumber \\ 
    & \left(\sum_{y_1=\tilde{y}_1}^{n_1} \binom{n_1}{y_1}\frac{\Gamma (y_1+\alpha ) \Gamma (n_1-y_1+\beta )}{B(\alpha ,\beta ) \Gamma (n_1+\alpha +\beta )}  v_1 \frac{\alpha+y_1}{\alpha+\beta+n_1} +\sum_{y_1=0}^{\tilde{y}_1-1} \binom{n_1}{y_1}\frac{\Gamma (y_1+\alpha ) \Gamma (n_1-y_1+\beta )}{B(\alpha ,\beta ) \Gamma (n_1+\alpha +\beta )}  v_2 \frac{\alpha+y_2}{\alpha+\beta+n_2}\right)
\end{align}
}
with
\[\tilde{y}_1 = \alpha\left(\frac{v_2}{v_1}\frac{\alpha+\beta+n_1}{\alpha+\beta+n_2}-1\right)+y_2\left(\frac{v_2}{v_1}\frac{\alpha+\beta+n_1}{\alpha+\beta+n_2}\right)\]
\end{enumerate} 
\label{prop:exp_profit_bb}
\end{prop}
\begin{proof}
To prove the first item, the expected profit in the test stage is:
\begin{equation}
    E[\pi_T] = \int_{p_1} \sum_{y_1=0}^{n_1} v_1 y_1 Pr(y_1|p_1) f(p_1) dp_1 + \int_{p_2} \sum_{y_2=0}^{n_2} v_2 y_2 Pr(y_2|p_2) f(p_2) dp_2
\end{equation}
 
Because $\sum_{y_j=0}^{n_j} y_j Pr(y_j|p_j) = n_j p_j$, then $\int_{p_1} \sum_{y_1=0}^{n_1} y_1 Pr(y_1|p_1) f(p_1) dp_1 = n_j \frac{\alpha}{\alpha+\beta}$, and plugging this in yields the expression in in the proposition.

The prove the second item, the \emph{a priori} expected profit in the roll stage is:
{\footnotesize
\begin{equation}
(N-n_1-n_2)\int_{p_2} \int_{p_1} \sum_{y_1=1}^{n_1} \sum_{y_2=1}^{n_2} \left[\delta(y_1,y_2)  p_1 v_1 + (1-\delta(y_1,y_2)) p_2 v_2 \right] Pr(y_2|p_2) Pr(y_1|p_1) f(p_1) f(p_2) dp_1 dp_2
\end{equation}
}

Focusing on the first additive (the second will be symmetric because of the symmetric prior), it can be written as:
\begin{equation}
\label{eq:bb-prof-dep}
(N-n_1-n_2) v_1 \int_{p_2} \int_{p_1} \sum_{y_1=1}^{n_1} \sum_{y_2=1}^{n_2} \delta(y_1,y_2)  Pr(y_2|p_2) Pr(y_1|p_1) p_1 f(p_1) f(p_2) dp_1 dp_2
\end{equation}

The optimal decision rule $\delta(y_1,y_2)$ is to pick the treatment with the highest expected posterior profit $v_j E[p_j|y_j]=v_j \frac{\alpha+y_j}{\alpha+\beta+n_j}$, resulting from the fact that the profits are Binomially distributed with a Beta prior.
Hence, by denoting $\tilde{y}_1 = \alpha\left(\frac{v_2}{v_1}\frac{\alpha+\beta+n_1}{\alpha+\beta+n_2}-1\right)+y_2\left(\frac{v_2}{v_1}\frac{\alpha+\beta+n_1}{\alpha+\beta+n_2}\right)$, and by applying Fubini's theorem, we can rewrite \eqref{eq:bb-prof-dep} as:
\begin{equation}
\label{eq:bb-prof-dep2}
(N-n_1-n_2) v_1 \sum_{y_2=1}^{n_2} \sum_{y_1=\tilde{y}_1}^{n_1}  \int_{p_2} Pr(y_2|p_2) f(p_2) dp_2 \int_{p_1} p_1 f(p_1) Pr(y_1|p_1) dp_1
\end{equation}

The derivation above assumes that if the expected posterior profit of both treatments is equal, then treatment $1$ is chosen as a tie-breaking rule. We will show that this tie-breaking rule does not change the result if we opt for another rule (e.g., pick treatment 2 if tied, or pick one randomly).

Using Bayes rule, $Pr(y_j|p_j)f(p_j) = Pr(y_j) f(p_j|y_j)$. This implies that:
\begin{align}
    \int_{p_2} Pr(y_2|p_2) f(p_2) dp_2 & = Pr(y_2) \\
    \int_{p_1} p_1 f(p_1) Pr(y_1|p_1) dp_1 & = Pr(y_1) \frac{\alpha+y_1}{\alpha+\beta+n_1}
\end{align}
The second equation stems from the fact that $f(p_1|y_1)$ is the pdf of a $Beta(\alpha+y_1,\beta+n_1-y_1)$ distribution.

We can calculate $Pr(y_j)$ as:
\begin{equation}
    Pr(y_j) = \int_{p_j=0}^1 Pr(y_j|p_j) f(p_j) dp_j =\binom{n_j}{y_j}\frac{\Gamma (y_j+\alpha ) \Gamma (n_j-y_j+\beta )}{B(\alpha ,\beta ) \Gamma (n_j+\alpha +\beta )}
\end{equation}

Plugging into \eqref{eq:bb-prof-dep2}, the total roll stage profit is:
{\footnotesize
\begin{align}
    & (N-n_1-n_2)  \sum_{y_2=1}^{n_2} \binom{n_2}{y_2}\frac{\Gamma (y_2+\alpha ) \Gamma (n_2-y_2+\beta )}{B(\alpha ,\beta ) \Gamma (n_2+\alpha +\beta )} \cdot \nonumber \\ 
    & \left(\sum_{y_1=\tilde{y}_1}^{n_1} \binom{n_1}{y_1}\frac{\Gamma (y_1+\alpha ) \Gamma (n_1-y_1+\beta )}{B(\alpha ,\beta ) \Gamma (n_1+\alpha +\beta )}  v_1 \frac{\alpha+y_1}{\alpha+\beta+n_1} +\sum_{y_1=0}^{\tilde{y}_1-1} \binom{n_1}{y_1}\frac{\Gamma (y_1+\alpha ) \Gamma (n_1-y_1+\beta )}{B(\alpha ,\beta ) \Gamma (n_1+\alpha +\beta )}  v_2 \frac{\alpha+y_2}{\alpha+\beta+n_2}\right)
\end{align}
}
If there is a tie such that $v_1 \frac{\alpha+y_1}{\alpha+\beta+n_1} = v_2 \frac{\alpha+y_2}{\alpha+\beta+n_2}$, it does not matter if we take the left or the right additive within the parenthesis. Hence, any tie-breaking rule will yield an equivalent profit.
\end{proof}

To design a test for binomial experiment, the expected profit from Proposition \ref{prop:exp_profit_bb} can be numerically optimized, using a discrete optimization heuristic. However, since the normal-normal model is more computationally convenient, it can be used to approximate the beta-binomial using the usual binomial approximation: $\mu=\frac{\alpha}{\alpha + \beta}$, $s=\sqrt{\mu(1-\mu)}$ and $\sigma=\sqrt{\frac{\alpha\beta}{(\alpha + \beta)^2 (\alpha + \beta + 1)}}$. Figure \ref{fig:nn_v_bb} shows that this approximation results in nearly the same sample size except when the response rate $\mu$ is close to zero (or equivalently close to 1) and the prior is relatively informative (prior precision = $\alpha + \beta > 100$) . 

\begin{figure}
\centering
\includegraphics[page=2, width=\textwidth]{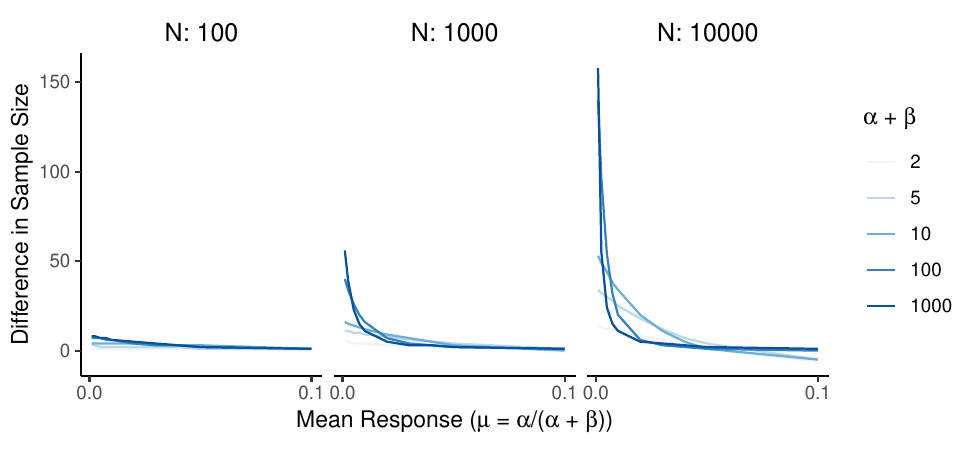}
\caption{Comparison of optimal sample sizes computed exactly using the Beta-Binomial model versus the Normal-Normal approximation for various values of mean response rate ($\mu = \frac{\alpha}{\alpha + \beta}$) and prior precision ($ \alpha + \beta$). }
\label{fig:nn_v_bb}
\end{figure}

\section{Asymmetric Tests}
\label{sec:asymmetric_test_details}

\subsection{Incumbent Challenger Test}
In an incumbent/challenger test more is known about one treatment than the other. Denote $\sigma_2 = c \sigma_1$ with $c>1$. To analyze this scenario in closed form, we will assume that $\mu_1=\mu_2$ and that $s_1=s_2=s$, although the solution can be found numerically for any set of values. Because the uncertainty is larger for treatment $2$, it is always the case that $n_2^*>n_1^*$ in an incumbent/challenger test. When the population size is small enough, it is too wasteful to experiment with treatment $1$, and the test will only include exposures to treatment $2$. After this test phase, comparison will be made to the prior on treatment $1$ to select which treatment to deploy. This is shown formally in Proposition \ref{prop:incumbent}: 

\begin{prop}[Incumbent/Challenger sample sizes]
\label{prop:incumbent}
In an asymmetric test when treatment $1$ is an incumbent and treatment $2$ is a challenger such that $\mu_1=\mu_2$, $s_1=s_2=s$ and $\sigma_2 = c \cdot \sigma_1$ with $c>1$:
\begin{enumerate}
    \item The optimal sample sizes are:
    \begin{align}
    n_1^* & = \frac{s \left(\sqrt{2 c^2 \left(c^2+1\right) N \sigma_1^2+\left(2 c^4+5 c^2+2\right) s^2}-c s(1+2 c^2)\right)}{2 \left(c^3+c\right) \sigma_1^2} \\
    n_2^* & = \frac{s \left(c \sqrt{2 c^2 \left(c^2+1\right) N \sigma_1^2+\left(2 c^4+5 c^2+2\right) s^2}-\left(c^2+2\right) s\right)}{2 c^2 \left(c^2+1\right) \sigma_1^2}.
\end{align}
    \item $n_2^*>n_1^*$ for any value of $N$, $s$, $c>1$ and $\sigma$.
    \item $n_2^*>0$, for any value of $N$, $s$, $c>1$ and $\sigma$. $n_1^*>0 \iff N>\frac{\left(2 c^4-c^2-1\right) s^2}{c^2 \sigma_1^2}$
\end{enumerate}
\end{prop}
\begin{proof}
Plugging $\mu_1=\mu_2=\mu$, $s_1=s_2=s$ and $\sigma_2 = c \sigma_1$ into the expected profit derived in Proposition \ref{prop:exp-prof}, the expected profit in an incumbent/challenger experiment is $N \mu + (N-n_1-n_2) \frac{v}{\sqrt{2 \pi}}$ with $v=\sqrt{\frac{\sigma_1^4}{\sigma_1^2+s^2/n_1}+\frac{c^4 \sigma_1^4}{c^2 \sigma_1^2+s^2/n_2}}$.

After simplifying and re-arranging, the first order conditions for finding the optimal $n_1$ and $n_2$ are:
\begin{align}
-\frac{s^2 (-N+n_1+n_2)}{\left(n_1 \sigma_1^2+s^2\right)^2} & = 2 \left(\frac{c^4 n_2}{c^2 n_2 \sigma_1^2+s^2}+\frac{n_1}{n_1 \sigma_1^2+s^2}\right) \\
\frac{c^4 s^2 (N-n_1-n_2)}{\left(c^2 n_2 \sigma_1^2+s^2\right)^2} & = 2 \left(\frac{c^4 n_2}{c^2 n_2 \sigma_1^2+s^2}+\frac{n_1}{n_1 \sigma_1^2+s^2}\right) \label{eq:foc}
\end{align}

Dividing the two equations and solving for $n_1$, the only possible solution such that $n_1^*>0$ for some values is $n_1^*=\frac{c^2 n_2 \sigma_1^2-c^2 s^2+s^2}{c^2 \sigma_1^2}$. Plugging this into Equation \eqref{eq:foc} yields the expression in the proposition, proving the first item.

To prove the second item, the inequality $n_2^*-n_1^*>0$ can be written as:
\begin{equation}
    (n_2^*-n_1^*)\frac{2\sigma_1^2 c^2 (1+c^2)}{s} = 2 \left(c^4-1\right) s >0
\end{equation}
which always holds because $c>1$.

To prove the third item, we solve for $n_2^*>0$, which holds for the described parameter values, and $n_1^*>0$ which holds if and only if $N > \frac{\left(2 c^4-c^2-1\right) s^2}{c^2 \sigma_1^2}$.
\end{proof}

\section{Thompson Sampling for the Normal-Normal model}
\label{sec:thompson_sampling_details}

Thompson sampling \citep{Thompson1933} has recently become the prominent heuristic for solving multi-armed bandit problems, due to its superior performance and ease of implementation \citep{scott2010, schwartzBradlowFader2016WP}. Here we describe the Thompson sampling algorithm we use, which is the standard implementation applied to the normal symmetric model.

Opportunities to apply the treatment are assumed to come in one at a time for each $i = 1 \ldots N$. Under the symmetric normal model, treatment $j$ generates outcomes $y_{ji}$ drawn from $\mathcal{N}(m_j,s^2)$. 

The algorithm is initialized with with priors $m_j \sim \mathcal{N}(\mu_j(0),\sigma_j^2(0))$. For each $i$, the algorithm makes a dynamic decision whether to deploy treatment $1$ or treatment $2$ as follows:
\begin{enumerate}
\item Draw a mean $m_1(i)$
from $\mathcal{N}(\mu_1(i-1), \sigma_1^2(i-1))$ and $m_2(i)$
from $\mathcal{N}(\mu_2(i-1), \sigma_2^2(i-1))$.

\item If $m_1(i)>m_2(i)$, treatment $1$ is deployed. Otherwise treatment 2 is deployed.

\item Either $y_{1i}$ or $y_{2i}$ is observed based on the decision. In simulation $y_{ji}$ is drawn from its true distribution $\mathcal{N}(m_j, s^2)$.

\item The hyperparameters $\mu_j(i)$ and $\sigma_j^2(i)$ are updated given the new data. If treatment $j$ was not deployed, the hyperparameters at time $i$ equal those at time $i-1$. If the treatment was deployed, the hyperparameters are calculated as the posterior of the normal distribution, with the observed outcome used as data and the hyperparameters from period $i-1$ used for the prior.

\end{enumerate}

Thus, treatments are probabilistically sampled according to the current probability that each treatment is best, i.e., treatment 1 is sampled at the rate of $Pr(\mu_1(i) > \mu_2(i))$. This rule favors treatments with higher expected response and, as a result, the algorithm will quickly converge to the best-performing treatment as data accumulates. However, it also is also more likely to sample treatments with higher uncertainty, because of the high potential upside for those treatments, which helps to avoid converging to the wrong treatment.  

The explicit explore versus exploit trade-off in a multi-armed bandit is similar to the tradeoff between the size of the test sample and the remaining population in a test \& roll, albeit more dynamic. The dynamic approach works better when opportunities to apply the treatment are spread out over time and the desired response is immediately available (e.g., website tests where the response is a click), but can be difficult to execute when the response is not immediately observable (e.g., sales) or when the treatments are sent out in batches (e.g., direct mail). 

\cite{agrawal2013} have shown that the regret from Thompson sampling with Normal outcomes and Normal priors is $O(\sqrt{N})$. This has been shown before to be the best achievable regret for any dynamic multi-armed bandit approach when compared to having perfect information, and hence Thompson sampling is an ideal benchmark for comparison.

\section{Application Details}
\label{sec:estimating_priors}

If a firm has data on response to prior marketing treatments that are similar to those that will be tested, this data can be used to estimate the distribution of mean response needed to compute the test \& roll sample size. For example, if the firm has past data on response $y_{ij}$ for each customer $i$ in each of $j = 1, \dots J$ previous marketing campaigns, then we can fit a hierarchical model: 
\begin{align}
 y_{ij} &\sim \mathcal{N}(m_j, s) \textnormal{ for observations } i = 1, \dots N_j \textnormal{ in campaigns } j = 1, \dots J\\
 m_j &\sim \mathcal{N}(\mu, \sigma) \textnormal{ in campaigns } j = 1, \dots J
 \label{eq:hb_model}
\end{align}
Estimates of $\mu$ and $\sigma$ can be plugged into \eqref{eq:NN_sample_size} to compute the test \& roll sample size. For binary responses, with small samples, we could estimate a similar beta-binomial model. 

The campaigns $j$ can be defined by a particular period of time when a marketing treatment was in place and the response was stable, such as response rates to direct marketing campaigns or customers visiting a website in a particular month. The key assumption is that these prior campaigns represent the range of likely mean responses for the treatments in the test that is being planned. We provide more details for specific applications below. 

\subsection{Website Testing Example} 
\label{sec:website-test-data}

The data on website tests is adopted from \cite{berman2018} and contains the results for 2,101 A/B tests. These tests were conducted across a wide variety of pages and websites. For each test we observe the number of times the page was served with each of the two variations and the total number of times a user clicked on the page for each variation. Our goal is to use this data to estimate the range of lifts in click rates that one might expect from a website test and then use this to size a test \& roll experiment. 

Figure \ref{fig:website_dist} displays the distribution of observed lift values between -0.6 and 0.6. This range contains 2,084, or 99.15\% of the experiments. The distribution is long tailed with a small number of experiments having higher lifts than 0.6. The interquantile [1\%, 99\%] lift range is [-.213, .327] with a mean of .112 and a median of .0015. For treatment effects, the range is [-.10, .16] with a mean of 0.005 and a median of 0.001. The sample sizes range from 100 to 17.4M, with an interquantile [1\%, 99\%] range of [116, 903,850], a mean of 574,474 and median of 3,864 users per treatment.

\begin{figure}[ht]
    \centering
    \includegraphics[width=0.7\textwidth]{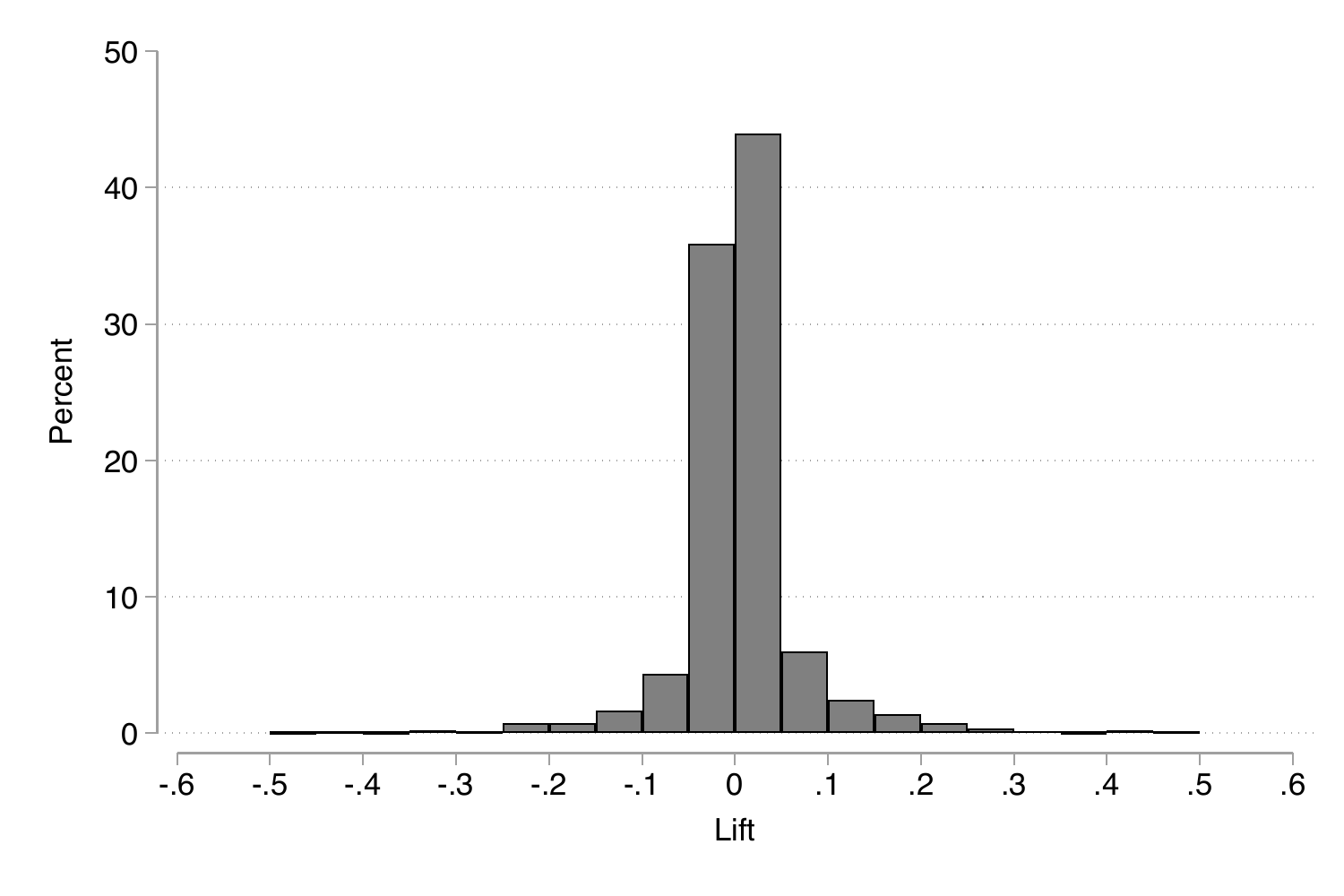}
    \caption{Distribution of lift values for experiments. Adopted from \cite{berman2018}}.
    \label{fig:website_dist}
\end{figure}

Because these tests were conducted across many websites with a wide range of click rates, there tends to be correlation in the click rate between the two arms in the same experiment. To account for this, we assume that each experiment $k$ has it's own mean click rate $t_k$ and assume that the means for the treatment arms within the experiment are distributed normal around the click rate for the experiment as follows:
\begin{align}
y_{ijk} &\sim \mathcal{N}(m_{jk}, s)\\
\label{eq:website_model}
m_{jk} &\sim \mathcal{N}(t_k, \sigma) \\
t_k &\sim \mathcal{N}(\mu, \omega)
\end{align}
Because the data is binary, we follow the binomial approximation and assume $s = m_{1k}(1-m_{1k})$, reducing the number of estimated population parameters to three. The model is estimated using the HMC algorithm implemented in Stan \citep{stan2018} with diffuse priors on the hyper-parameters and the estimates are reported in Table \ref{tab:website_model_estimates}.
\linespread{1.1}
\begin{table}[ht]
\caption{Model estimates for meta-analysis of website tests provide an estimate of the distribution of mean response to be used in planning future tests with similar treatments and targeted populations.}
\label{tab:website_model_estimates}
\centering
\begin{tabular}{ccccc}
\hline
 & mean & sd & 2.5\%-tile & 97.5\%tile\\
\hline
$\mu$ & 0.676 & 0.004 & 0.667 & 0.685\\
$\sigma$ & 0.030 &  0.001 & 0.029 & 0.031\\
$\omega$ & 0.199 & 0.003 & 0.193 & 0.206\\
\hline
\end{tabular}
\end{table}
\linespread{1.5}

In the empirical model, $\omega$ captures the variation in mean response \emph{across} experiments, while $\sigma$ captures the variation between arms \emph{within} an experiment. In sizing a test \& roll experiment following \eqref{eq:NN_sample_size}, we are interested in the potential differences between arms within a single experiment, so we use the estimate of $\sigma$ and ignore $\omega$. In addition, we assume $s=\mu(1-\mu)$, but if the experimenter has other information about the likely click rate for these particular web pages, then $s$ can be appropriately adjusted or conservatively set at 0.25, while still using $\sigma$ as an estimate of the range of mean responses expected for treatments within an experiment. 

\subsection{Display Ad Testing Example}

We illustrate how ``Advertiser 1" in \citet{LewisRao2015} might obtain the parameters $\mu$ and $\sigma$ in order to find the profit-maximizing sample size for a new test \& roll with treatments that are expected to perform similarly to experiments 1.1, 1.2, 1.3, 1.5 and 1.6 reported in Table \ref{tab:display_test_data}.  We eliminated experiment 1.4 because it had a substantially different media cost and response rate for the control group versus the other experiments and appears to be targeting customers with higher baseline purchase propensity.

\linespread{1.1}
\begin{table}[ht]
\caption{Reported mean and standard deviation for control group in display advertising tests reported by \citet[Table I]{LewisRao2015}. These are used to estimate mean and variance for  display advertising response for ``Advertiser 1"'s typical campaign}
\label{tab:display_test_data}
\centering
\begin{tabular}{c c c c}
\hline
Test & Mean ($\bar{y}_j$) & Pooled St. Dev  ($\hat{s}$) & Group Size ($n$)\\
\hline
1.1 & 9.49 & 94.28 & 300,000 \\
1.2 & 10.50 & 111.15 & 300,000\\
1.3 & 4.86 & 69.98 & 300,000\\
1.5 & 11.47 & 111.37 & 300,000\\
1.6 & 17.62 & 132.15 & 300,000\\
\hline
%Mean & 10.79 & 103.77\\
%St. Dev. & 4.58 \\
%\hline
\end{tabular}
\end{table}
\linespread{1.5}

Using the data in Table \ref{tab:display_test_data}, we estimate the following hierarchical model for the mean response in the control group reported for each experiment $j$. 
\begin{align}
\bar{y}_{j} &\sim \mathcal{N}\left(m_{j}, \frac{\hat{s}}{\sqrt{n}}\right)\\
m_{j} &\sim \mathcal{N}(\mu, \sigma)
\end{align} 
where the sampling distribution for $y_{ijk}$ in equation \ref{eq:website_model} has been replaced with $\bar{y}_j$, since we do not have access to the user-level data. The estimates of $\mu$ and $\sigma$ reported in Table \ref{tab:display_test_metaanalysis} are used in designing a new test \& roll for Advertiser 1. $s$ is estimated as the average of $s_j$ across the 5 experiments, which is 103.77. 

\linespread{1.1}
\begin{table}
\centering
\caption{Model estimates for meta-analysis of display advertising tests provide an estimate of the distribution of mean response to be used in planning future tests with similar treatments and targeted populations.}
\label{tab:display_test_metaanalysis}
\begin{tabular}{ccccc}
\hline
Parameter & mean & sd & 2.5\%-tile & 97.5\%-tile \\
\hline
$\mu$ & 10.36 & 1.99 & 6.16 & 14.17\\
$\sigma$ & 4.40 & 1.17 & 2.63 & 7.17\\
\hline
\end{tabular}
\end{table}
\linespread{1.5}

Because we are estimating the variance in mean response $\sigma$ from just 5 experiments, the posterior of $\sigma$ is relatively wide.  As can be seen from \eqref{eq:NN_sample_size}, the profit-maximizing sample size will be largest when $\sigma$ is smallest.  Conservatively, one might use the posterior 2.5\%-tile for $\sigma$, instead of the posterior mean. This results in a profit-maximizing sample size of 18,486, still far smaller than that recommended for a hypothesis test.

\subsection{Catalog Holdout Testing Example}
\label{sec:catalog_details}

The catalog holdout data describe 30 catalog holdout tests conducted in October 2013 through March 2014. In each month, 6 tests were conducted using the same print catalog and different targeted populations. This data is provided by the same retailer as in \citet{ZantedeschiFeitBradlow2017}, but is a different sample of tests. For each customer $i$ in each test $k$, we observe the all-channel sales $y_{ijk}$ for one month after the catalog is delivered. The sample sizes and holdout rates for these tests vary with a [10\%, 90\%] interquartile range for the sample size of [346.7, 1395.5] with a mean of 658.3 and a median of 437.0. The holdout rates also varied widely with a range of [1.1\%, 95.1\%], a median of 5.4\% and a mean of 21.0\%. 

The distribution of the estimated treatment effects are shown in Figure \ref{fig:catalog_tes}. The interquartile range for the point estimates of the treatment effects is [-10.77, 39.15] with a median of 6.23 and a mean of 11.34 (all \$US). Lifts can not be computed for 7 of the tests because no purchases were made in the control group, but the median lift is 1.48 and the 10th percentile is -0.549.  The 1-month purchase amounts for individual customers have a median of 0, a mean of 43.64 and a 90th percentile of 113.00. 

\linespread{1.1}
\begin{figure}[ht]
    \centering
    \includegraphics[page=2, width=0.45\textwidth]{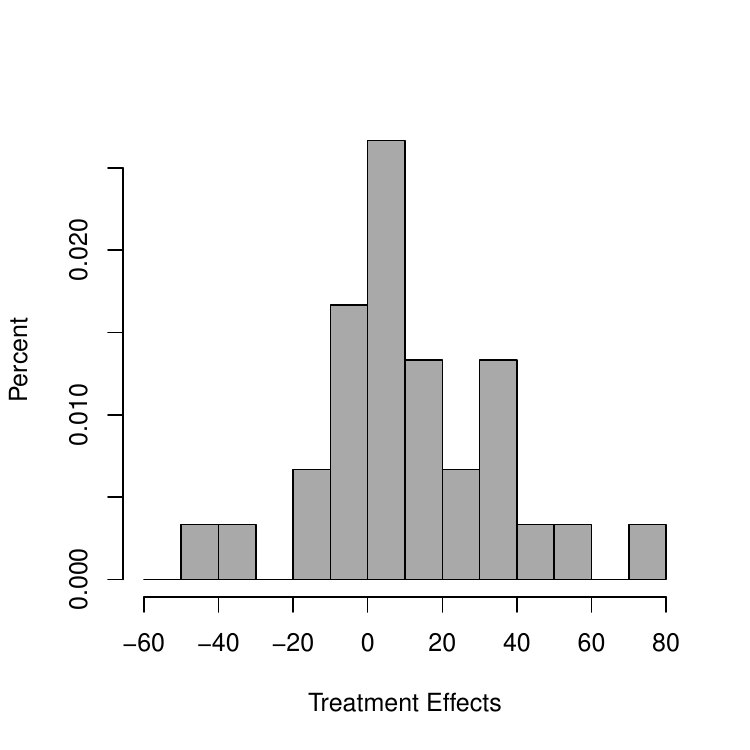}
    \caption{Distribution of catalog holdout test treatment effects ($\overline{y}_2 - \overline{y}_1$)}
    \label{fig:catalog_tes}
\end{figure}
\linespread{1.5}

Individually, the catalog holdout tests have very imprecise estimates for response due to small sample size and high noise in the data. The hierarchical model is particularly valuable in pooling information across the tests and propagating uncertainty due to small sample sizes. We fit a model similar to that used for the website tests, except that we allow for $\mu_1 \ne \mu_2$ and $\sigma_1 \ne \sigma_2$, because unlike for the website tests, there is a clear distinction between the treated and holdout conditions. The model we fit is: 
\begin{align}
 y_{i1k} &\sim \mathcal{N}(m_{1k}, s_1)  \textnormal{ for customers in control group} \\
 y_{i2k} &\sim \mathcal{N}(m_{2k}, s_2)  \textnormal{ for customers in treatment group} \\
 m_{1k} &\sim \mathcal{N}(t_k, \sigma_1)  \\
 m_{2k} &\sim \mathcal{N}(t_k + \Delta, \sigma_2)\\
 t_k &\sim \mathcal{N}(\mu_1, \omega)
\end{align}
By modeling the overall response rate for the experiment $t_k$, we allow for the different targeted populations to have different response rates and account for the correlation in response within experiments. In planning a new test, we focus on the the variation in response rates within the experiment, as estimated by $\sigma_1$ and $\sigma_2$. 

Samples from the posterior are obtained using the HMC algorithm implemented in Stan with uniform priors on the hyper-parameters. The posterior means for $\mu_1$, $\Delta$ = $\mu_2 - \mu_1$, $\sigma_1$ and $\sigma_2$ reported in Table \ref{tab:catalog_model_estimates} are used as point estimate to compute the asymmetric test \& roll sample size. 
\linespread{1.1}
\begin{table}[ht]
\centering
\caption{Model estimates for meta-analysis of catalog holdout tests provide an estimate of the distribution of mean response to be used in planning future tests.}
\label{tab:catalog_model_estimates}
\begin{tabular}{ccccc}
\hline
 & mean & sd & 2.5\%-tile & 97.5\%tile\\
\hline
$s_1$ & 87.69 & 1.21 & 85.41 & 90.06\\
$s_2$ & 179.36 & 0.97 & 177.46 & 181.24\\ 
$\mu_1$ & 19.39 & 7.13 & 5.32 & 33.17\\
$\Delta$ & 10.67 & 6.19 & -1.30 & 22.77 \\
$\sigma_1$ & 20.97 & 5.85 & 8.81 & 32.25\\
$\sigma_2$ & 13.48 & 5.88 & 4.01 & 26.78\\
$\omega$ & 27.25 & 5.18 & 18.27 & 38.57\\
\hline
\end{tabular}
\end{table}
\linespread{1.5}

We also estimated a version of the model where $s_1$ was constrained to be the same as $s_2$ and used these estimates to show that unequal group sizes can arise from the priors (unlike in null hypothesis testing). The resulting estimates are reported in Table \ref{tab:catalog_model_estimates_2}.
\linespread{1.1}
\begin{table}[ht]
\caption{Model estimates for catalog holdout tests assuming $s_1 = s_2 = s$.}
\label{tab:catalog_model_estimates_2}
\centering
\begin{tabular}{ccccc}
\hline
 & mean & sd & 2.5\%-tile & 97.5\%tile\\
\hline
s & 170.17 & 0.86 & 168.48 & 171.91\\
$\mu_1$ & 23.32 & 8.38 & 6.64 & 40.06\\
$\Delta$ & 6.39 & 7.48 & -7.78 & 21.93\\
$\sigma_1$ & 18.54 & 7.08 & 5.55 & 33.62\\
$\sigma_2$ & 9.79 & 5.98 & 2.37 & 23.73\\
$\omega$ & 28.75 & 4.87 & 20.20 & 39.15\\
\hline
\end{tabular}
\end{table}
\linespread{1.5}

%\section{Regret Sensitivities}
%\label{sec:regret_simulation}

\end{appendices}

\end{document}